\documentclass[10pt,doublesided,doublecolumn]{IEEEtran}
\usepackage{graphicx,epsfig,amsmath,amssymb,bm}
\usepackage{amssymb,balance}
\usepackage{amsmath}
\usepackage{amsthm} 
\usepackage{tikz}
\usetikzlibrary{calc}
\usepackage{pgfplots}
\usetikzlibrary{plotmarks}
\usepackage{amssymb}
\usepackage{psfrag}
\usetikzlibrary{decorations.shapes} 

\newtheorem{lemma}{Lemma}

\newtheorem{theorem}{Theorem}
\newtheorem{algorithm}{Algorithm}
\def\LSB{\left[}        
\def\RSB{\right]}       
\def\LB{\left(}         
\def\RB{\right)}        

\newfont{\bbb}{msbm10 scaled 500}

\newfont{\bb}{msbm10 scaled 1100}
\newcommand{\CC}{\mbox{\bb C}}
\newcommand{\RR}{\mbox{\bb R}}

\newcommand{\bv}{{\bf b}}

\newcommand{\hv}{{\bf h}}

\newcommand{\wv}{{\bf w}}
\newcommand{\vv}{{\bf v}}
\newcommand{\xv}{{\bf x}}

\newcommand{\Am}{{\bf A}}
\newcommand{\Bm}{{\bf B}}

\newcommand{\Em}{{\bf E}}

\newcommand{\Hm}{{\bf H}}
\newcommand{\Id}{{\bf I}}

\newcommand{\Pm}{{\bf P}}
\newcommand{\Qm}{{\bf Q}}

\newcommand{\Wm}{{\bf W}}

\newcommand{\Xm}{{\bf X}}

\newcommand{\trace}{{\hbox{tr}}}

\newcommand{\defines}{{\,\,\stackrel{\scriptscriptstyle \bigtriangleup}{=}\,\,}}

\newtheorem{proposition}[theorem]{Proposition}
\newtheorem{corollary}[theorem]{Corollary}

\newcommand{\beqa}{\begin{eqnarray}}
\newcommand{\eeqa}{\end{eqnarray}}
\newcommand{\dsp}{\displaystyle}
\def\argmax{\operatornamewithlimits{arg\,max}}

\begin{document}
%

\title{Transmit Power Minimization in Small Cell Networks Under Time Average QoS Constraints}
\author{Subhash~Lakshminarayana,~\IEEEmembership{Member, IEEE}, Mohamad Assaad~\IEEEmembership{Member, IEEE} and M\'erouane Debbah~\IEEEmembership{Fellow, IEEE} 
\thanks{Manuscript received 21 July, 2014; revised 15 Dec, 2014; accepted 21 Feb, 2015. 
This paper was presented in part at the IEEE International Workshop on Signal Processing Advances in Wireless Communications, SPAWC 2013. 
S. Lakshminarayana is with the Advanced Digital Sciences Center, Illinois at Singapore, Singapore 138632 (e-mail: subhash.l@adsc.edu.sg). M.Assaad is with the ``Laboratoire des Signaux et Systemes (L2S, UMR CNRS 8506), CENTRALESUPELEC", France. M. Debbah is with CENTRALESUPELEC. e-mail: subhashl@princeton.edu,  mohamad.assaad@centralesupelec.fr and  merouane.debbah@centralesupelec.fr.

M\'erouane Debbah was supported by the European Research Council under Grant 305123 MORE.
The research of M. Assaad was partially supported by the Celtic Project ``SHARING."
}
}

\maketitle
\begin{abstract}
We consider a small cell network (SCN) consisting of 
$N$ cells, with the small cell base stations (SCBSs) equipped with $N_t \geq 1$ antennas each, serving $K$ single antenna user terminals (UTs) per cell.
Under this set up, we address the following question: given certain time average quality of service (QoS)
targets for the UTs, what is the minimum transmit power expenditure with which they 
can be met? Our motivation to consider time average QoS constraint comes from the fact that modern wireless applications such as file sharing, multi-media etc. allow some flexibility in terms of their delay tolerance.
Time average QoS constraints can lead to greater transmit power savings as compared to instantaneous QoS constraints since it provides the flexibility to dynamically allocate resources over the fading channel states.
We formulate the problem as a stochastic optimization problem whose solution is the design of the downlink beamforming vectors during each time slot.
We solve this problem using the approach of Lyapunov optimization and characterize the performance of the proposed algorithm.
With this algorithm as the reference, we present two main
contributions that incorporate practical design 
considerations in SCNs.
First, we analyze the impact 
of delays incurred in information exchange between the SCBSs.
Second, we impose channel state information (CSI) feedback constraints, and formulate a joint CSI feedback and beamforming strategy.
In both cases, we provide performance bounds of the algorithm in terms of satisfying the QoS constraints and the time average power expenditure.
Our simulation results show that solving the problem
with time average QoS constraints provide greater savings in the transmit power as compared to the 
instantaneous QoS constraints.
\end{abstract}


\section{Introduction}
The global carbon emission related to information
and communication technology (ICT) has been increasing at an alarming rate (an increase of 10\% every year) \cite{GreenTouch2011,GeSI2020}. 
Therefore, energy efficiency is becoming an important concern in the design of future wireless networks both from environmental and economical point of view.
Minimizing the transmit power leads to
a significant reduction in the overall power consumption
at the base stations, hence leading to greater energy efficient design \cite{EnergyEff2010}.

At the same time, there is an ever growing demand for higher data rate services and quality of service (QoS) guarantees among mobile users. Network densification has been identified as a promising solution to satisfy
the growing data rate demands, resulting in massive deployment of 
small cell (SCs) leading to greater spatial reuse \cite{AndrewsClaussen2012}, \cite{Hoydis2010}. Nevertheless, SCs alone cannot provide seamless coverage over large areas, and hence they must co-exist with the macro-base station, resulting in heterogeneous networks (HetNets).
The deployment of SCs must be planned carefully (so that they can co-exist with macro-cellular networks), and hence low complexity, decentralized interference management schemes are very important in HetNet design \cite{RanganFemtocells2010}, \cite{XiaAndrewsFemto2010}.

Motivated by the aforementioned developments, in this work
we consider the problem of minimizing the transmit power in 
SCNs, in which the small cell base stations (SCBSs) are equipped with multiple antennas. In HetNets, while the macro-base stations  mainly provide coverage and signalling information, the  data intensive applications such as file sharing, multi-media etc. are transmitted by the SCBSs. 
Additionally, such applications also allow some latitude in terms
of delay tolerance \cite{HPVideoStreaming2002}. Motivated by this fact, we consider the problem of minimizing the transmit power expenditure at the SCBSs
subject to time average QoS constraints of the user terminals (UTs), where the QoS constraints have to be met over a long period of time (and not instantaneously). Time average QoS constraints provide the flexibility to dynamically allocate resources over the fading channel states
as compared to instantaneous QoS constraints.  
In terms of energy savings, time average QoS constraint can lead to better performance, due to the fact that the 
transmissions can be delayed until favourable channel conditions are seen, thus minimizing the
energy expenditure.
The concept has also been exploited in the context of \emph{energy-delay trade offs} \cite{UysalPrabhakarGamal2002,BerryGall2002}.

\subsection*{Related Work}
The issue of minimizing the downlink power and beamforming design in multi-antenna systems 
subject to UT signal to interference noise ratio (SINR) constraints was first addressed in 
\cite{RashidTassiulas1998}. The authors proposed
an iterative algorithm based on uplink-downlink duality that
converges to a feasible solution (if the solution exists). 
The aspect of feasibility of the downlink SINR targets,
and the design of optimal beamforming vectors to minimize
the transmit power was studied for a multi-user downlink scenario
in \cite{SchubertBoche2004} and \cite{SchubertBoche2005}.
The downlink power minimization problem was solved using a second order conic programming
(SOCP) based approach in \cite{WeiselShamai2006}. 
This result was interpreted in a Lagrangian duality based   
framework for the single cell
and multi-cell scenario in \cite{YuL07} and \cite{Dahrouj2010}
respectively. 
References \cite{SayedLeakage2007}, \cite{DingLeakageTSP2013} and \cite{DingLeakageTWC2014}
proposed precoder design that maximize the so-called signal-to-leakage-and-noise ratio (SLNR) for all UTs simultaneously.
However, all the aforementioned works 
consider the problem of beamforming design during a given time slot, for a fixed channel realization and instantaneous 
QoS constraints. 

The problem of handling delay optimal precoder design 
in multi-user MIMO systems has been addressed in \cite{LauChenTWC2009},\cite{HuangLauLimFB2012} using the approach of Markov decision problem.
However, these works are restricted to a single BS scenario, and do not address the issue of decentralized design, delayed information exchange, and CSI
feedback constraints that are very essential in SCN design.

\subsection*{Summary and Contributions}
In this work, we consider a stochastic version of the beamforming design problem in SCNs, and consider minimizing the transmit power subject to the time average QoS constraints.
We formulate our problem as a stochastic optimization problem,
and propose a solution based on the technique of Lyapunov optimization \cite{neelybook2006}, \cite{neelybook2010}.
The Lyapunov optimization technique provides simple online solutions based only on the current knowledge of the
system state (as opposed to traditional approaches such as dynamic programming which suffer from very high complexity and require a-priori 
knowledge of the statistics of all the random processes in the system). To the best of our knowledge, the application of 
Lyapunov optimization technique in the context of beamforming in MIMO systems under time average 
QoS constraints is novel.

We model the time average QoS constraints as \emph{virtual queues}, and transform the problem into a transmit power minimization problem under the queue stability constraint. We then use the technique of Lyapunov optimization 
to formulate the beamforming design during each time slot. 
We provide the performance bounds of this algorithm
in terms of satisfying the QoS constraints and the time average transmit power.
In our algorithm, the SCBSs can formulate the beamforming vectors using only the local channel state information (CSI).
The SCBSs would only have to exchange virtual queue-length information among themselves. 

We then present two main contributions in the context of Lyapunov optimization that incorporate practical design considerations in SCNs. \\
{\bf Delays incurred in information exchange between the SCBSs}  \\
We introduce delays in information exchange between the SCBSs (for e.g., delays introduced in the backhaul links), and 
characterize the performance of our algorithm under this scenario. Specifically, we show that the delays in information exchange among the SCBSs result in only a constant 
gap with respect to the performance of the case with no delays. Moreover,  under some conditions the gap vanishes, and the impact of delay on our algorithm becomes negligible.  \\
{\bf Limited CSI at the SCBSs} \\
Secondly, in order to limit the CSI feedback load,
we consider the case when the SCBS can obtain the CSI from
a limited number of UTs during each time slot.
In this case, we solve the problem of joint CSI feedback and beamforming design by using the Lyapunov optimization framework.

In practice, (e.g. in long term evolution (LTE) and LTE-advanced networks), obtaining the CSI feedback from all the UTs in the network becomes impractical as the number of UTs increase, since this leads to a huge feedback overhead (this is even true in a network consisting of single antenna links where the feedback is a simple scalar). 
Therefore, the SCBS have to decide which subset of UTs must feed back their CSI during each time slot. The problem of joint CSI feedback and transmission is known to be challenging. The main complexity lies in the fact that the transmitter must decide which UTs have to feedback their CSI without a-priori knowledge of their current channel
states \cite{GopCarShakkottai2007}. 
Furthermore, in MIMO systems, CSI knowledge is crucial in
formulating the beamforming vectors. 
The issue of joint CSI feedback and beamforming design problem (in terms of selection which UTs have to feedback their CSI) is very challenging.
Most of existing works in this context either assume a predefined beamforming strategy (e.g. ZF), and/ or let all the UTs feedback a quantized version of their CSI \cite{LoveHeathLimitedFB2008}, \cite{KobJindalCaire2011}, \cite{HuangLauLimFB2012}, \cite{DestounisMoDebbah2014}.
However, in practice, even when the CSI feedback is quantized, with a large number of UTs, only a subset of UTs can feed back their CSI (and choosing the optimal subset is complex).

For the case of joint CSI feedback and beamforming design, we provide the following results:
\begin{itemize}
\item We first present 
the feedback decision rule, and the corresponding beamforming design  strategy obtained by the analysis of Lyapunov optimization.
\item Second, we present a low complexity 
algorithm named $AlgF$ in order to optimally
solve the CSI feedback decision problem obtained by the analysis of Lyapunov optimization.
\item We then provide a performance analysis of the proposed algorithm, and compare it to the performance of the 
optimal solution. We also show that under certain settings the performance can be made very close to the optimal value. 
\end{itemize}

The rest of the paper is organized as follows.
We specify the system model in Section \ref{sec:sysmodel}.
We first provide the solution based on Lyapunov optimization
with perfect CSI at the SCBSs in Section \ref{sec:PerfectCSI}.
The impact of delayed information exchange between the SCBSs
is addressed in Section \ref{sec:DelInfoExchg}.
The case with limited CSI at the SCBS and the problem formulation in case of joint CSI feedback and beamforming design is considered in
Section \ref{sec:FB_BF_Opt}. The numerical results are provided in Section \ref{sec:NumResults}. Finally, 
the paper is concluded in Section \ref{sec:Conclusion}.
The technical proofs of the results in this paper are provided in Appendices A,B, and C.

\emph{Notations:} Throughout this work, we use boldface lowercase and
uppercase letters to designate column vectors and matrices,
respectively. For a vector $\xv$, $\xv^T$ and $\xv^H$ 
denote the transpose, the complex conjugate respectively (and similarly for a matrix).
The notation $\trace(\Xm),$ and $\lambda^{\max}(\Xm)$ 
denote the trace and the maximum eigenvalue of 
the matrix $\Xm$ respectively. We denote the identity matrix
by the notation $\Id.$

\section{System Model}
\label{sec:sysmodel}
\begin{figure}[htp]
\begin{center}
\begin{scriptsize}
 \includegraphics[width= 2.5 in]{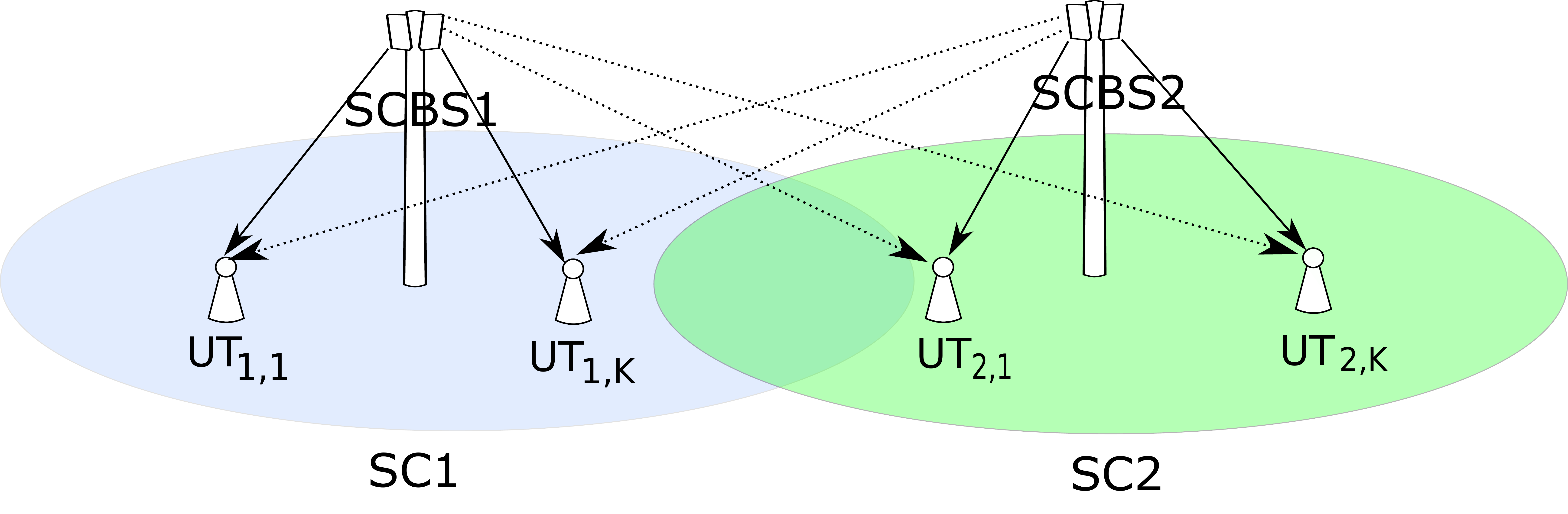}   
 \caption{Example of a SCN consisting of $2$ small cells.}
 \label{fig:SCN}
  \end{scriptsize}
  \end{center}
 \end{figure} 
The system model is illustrated in Figure \ref{fig:SCN}. We consider SCN consisting of $N$ cells and $K$ user terminals (UTs) per cell \footnote{In Subsection II.B, we provide extensions of the system model to the case of a multi-tier network.}. The SCBSs are equipped with $N_t$ antennas each, and the UTs have a single antenna. The notation UT$_{i,j}$ denotes the $j$-th UT present in the $i$-th SC. The SCBS of each cell serves only the UTs present in its cell. 
We consider a discrete-time block-fading channel model where the channel
remains constant for a given coherence interval and then changes
independently from one block to the other. We index the time slots 
by $t.$ We denote the channel vector from the SCBS$_i$  to the UT$_{j,k}$ 
during the time slot $t$ by $\hv_{i,j,k}[t] \in \CC^{N_t}.$
The elements of the channel vector are independent and identically distributed (i.i.d.).
Furthermore, we consider that the channel vector can be represented
by two components, i.e, $\hv_{i,j,k}[t] = \sigma_{i,j,k} \hv^{\prime}_{i,j,k}[t], $ where $\sigma_{i,j,k}$ is the pathloss component
and $\hv^{\prime}_{i,j,k}[t]$ is the fast fading component, and
$\mathbb{E}[\hv^{\prime}_{i,j,k}[t] (\hv^{\prime}_{i,j,k}[t])^H] = \Id$.
We define the channel matrix $\Hm[t]$ given
$\Hm_{i,j}[t]  = \LSB \hv_{i,j,1}[t],\dots,\hv_{i,j,K}[t] \RSB;$
$\Hm_{i}[t]  = \LSB \Hm_{i,1}[t],\dots,\Hm_{i,N}[t]\RSB$ and $\Hm[t]  = \LSB \Hm_{i}[t],\dots,\Hm_{N}[t]\RSB.$
The channel process $\{ \Hm[t], t = 0,1,2,\dots \}$ is assumed to 
be an i.i.d. (across time slots) discrete time stationary ergodic random process.
We also consider a practical assumption that the channel gains are bounded at all times, i.e., $|\hv_{i,n,k}[t]|^2 \leq H_{\max} \ \forall i,n,k.$
Throughout this work, we use the following definitions:
\begin{itemize}
\item \emph{Local CSI at BS$_i$}  : The CSI from BS$_i$ to 
all the UTs in the system, i.e., 
$\hv_{i,n,k}, \ n = 1,\dots,N, \ k = 1,\dots,K.$ 
\item \emph{Non-local CSI at BS$_i$} :  CSI from BS$_j$  ($j \neq i$) to UT$_{n,k}, \ n = 1,\dots,N, \ k = 1,\dots,K.$  
\item \emph{Global CSI} : CSI from all the BSs to all the UTs in the system. 
\end{itemize}
The local CSI corresponds to the information that can be obtained locally (either through feedback or uplink pilots in schemes such as the TDD), where as the non-local CSI is the information
that must be exchanged between the BSs.

Let us denote the beamforming vector corresponding to UT$_{i,j}$ during  slot $t$ by $\wv_{i,j}[t] \in \CC^{N_t}.$  
The signal received by UT$_{i,j}$ during time $t$ is given by
\begin{align}
y_{i,j}[t]  = {\hv}^H_{i,i,j}[t] \wv_{i,j}[t]x_{i,j}[t] + \sum_{\substack{(n,k) \\ \neq (i,j)}} \hv^H_{n,i,j}[t] & \wv_{n,k}[t] x_{n,k}[t] \nonumber \\ & +z_{i,j}[t],
\end{align}
where $x_{i,j}[t] \in \CC$ represents the information signal for the UT$_{i,j}$
during the time slot $t$ and $z_{i,j}[t]$
is the noise with variance $N_0.$
The downlink SINR for UT$_{i,j}$ is given by
\begin{align}
\text{SINR}_{i,j}[t] = \frac{|\wv_{i,j}[t] h_{i,i,j}[t]|^2}{\sum_{(n,k) \neq (i,j)} |\wv_{n,k}[t] h_{n,i,j}[t]|^2+N_0},
\end{align}
where the numerator represents the useful signal and the denominator terms represent the interference and the noise terms respectively. 
In the rest of this paper, we set $N_0 = 1,$ and normalize the useful signal and the interference signal 
power by the variance of the noise. Although not explicitly mentioned, henceforth all the signal powers 
are assumed to be the normalized values, and $N_0$ is taken to be $1.$
The transmit power of SCBS$_i$ denoted as $P_i[t]$ depends on the beamforming vector 
during the time slot $t$ which can be given as
$P_i[t] = \sum^K_{j = 1} \wv^H_{i,j}[t]\wv_{i,j}[t], \ i = 1,\dots,N.$
The optimization problem to minimize the average energy expenditure 
subject to time average QoS constraint can be formulated as 
\beqa
 &\dsp \min & \lim_{T \to \infty} \frac{1}{T}\sum^{T-1}_{t = 0}\mathbb{E} \Big{[} \sum^N_{i = 1} P_i[t] \Big{]} \label{eqn:OPT_basicEE}\\
& s.t. &  \lim_{T \to \infty} \frac{1}{T} \sum^{T-1}_{t = 0} \mathbb{E} \LSB {\gamma}_{i,j}[t] \RSB  \geq \lambda_{i,j}, \ \ \forall i,j \label{eqn:const1} \\
&& \sum^K_{j = 1} \wv^H_{i,j}[t]\wv_{i,j}[t] \leq P_{\text{peak}} \qquad \forall i,t ,\label{eqn:const11} 
 \eeqa
where $P_{\text{peak}}$ is the peak power at which the SCBSs can transmit.
${\gamma}_{i,j}$ is the QoS metric (to be specified) and $\lambda_{i,j}$ is the QoS target.
We define QoS metric as follows:
\begin{align}
\gamma_{i,j}[t]  = |\wv^H_{i,j}[t] & \hv_{i,i,j}[t]|^2 \nonumber \\ & - \nu_{i,j} \big{(}\sum_{\substack{(n,k) \\ \neq (i,j)}}|\wv^H_{n,k}[t]\hv_{n,i,j}[t]|^2+N_0 \big{)} \label{eqn:QoS2constraint},
\end{align}
where $\nu_{i,j} \geq 0$ is a scaling factor.
This QoS metric ensures that the time average useful signal power is greater than the time average
interference signal power by a threshold $\lambda_{i,j}$.

The optimization problem described by \eqref{eqn:OPT_basicEE}, \eqref{eqn:const1}, \eqref{eqn:const11} is a stochastic optimization problem.
The control action to be taken during each time slot is the formulation of the downlink beamforming
vectors ($\wv_{i,j}[t] \ \forall i,j$) during every time slot $t.$ 
Let $P_{\text{inf}}$ be the minimum time average
power incurred over all possible sequence of control actions 
that satisfy the constraints \eqref{eqn:const1}-\eqref{eqn:const11}.
The problem in \eqref{eqn:OPT_basicEE} - \eqref{eqn:const11} can be solved
optimally using techniques such as dynamic programming. But these methods
are computationally complex and suffer from the \emph{curse of dimensionality}.

We propose to solve this by the method of 
Lyapunov optimization. 
Although this method is sub optimal, 
the control actions using this method can be computed with the knowledge 
of only the current state of the 
system and does not require a-priori 
knowledge of the statistics of the random processes associated with the system.
Moreover, it has low computational complexity as compared to dynamic programming based
techniques.
The application
of this method transforms the stochastic optimization problem into a series of successive instantaneous static optimization problems. Convexity of these instantaneous
optimization problem is desirable to provide an efficient 
solution. However, the instantaneous optimization 
problems are not necessarily convex and depend on the
form of time average QoS metric (as we shall see later in Section \ref{sec:PerfectCSI}).

\subsection{Comment on the QoS metric}
The QoS metric chosen in this work represents the difference between time 
average useful signal power and the time average scaled interference signal power,
which is constrained to be above a threshold value.
One can view the QoS metric in similar spirit with metrics such as interference temperature control \cite{Gastpar2007,ZhangCui2009, Zhang2009,NekoueiDey2012} (in which the interference, peak interference power (PIP) or average interference power (AIP) is constrained to be below a certain threshold). Our QoS metric is more general in the sense that we constrain the interference signal to be below the useful signal power upto a
threshold level.

A QoS metric of more practical importance would be the time average SINR,
and the associated constraint given by
\begin{align}
\lim_{T \to \infty} \frac{1}{T} \sum^{T-1}_{t = 0} \mathbb{E} \LSB \text{SINR}_{i,j}[t] \RSB \geq \lambda_{i,j} \label{eqn:SINRConstraint}.
\end{align}
However, the QoS metric chosen in \eqref{eqn:QoS2constraint} (the difference form) has certain functional advantages
in our algorithm design, which will be specified in the rest of 
this subsection. Before enlisting them, we point out that with the help of numerical results, we will illustrate (in Section \ref{sec:NumResults}) that the time average SINR constraint of \eqref{eqn:SINRConstraint} is indeed satisfied by the algorithm developed in this work. Therefore, it can be applied directly to problem with time average SINR constraint.

The main drawback of using the time average SINR
constraint (in our algorithm design) is that the application of Lyapunov optimization technique with this metric leads to solving successive instantaneous  optimization problems that are non-convex (as we shall see in the algorithm formulation in Section \ref{sec:PerfectCSI}). 
Consequently, we can only find a suboptimal solution 
to these non-convex problems. Furthermore, it is hard to characterize the gap between 
the solution provided by the Lyapunov optimization technique, and 
$P_{\text{inf}}.$
It is also hard to find a decentralized implementation 
of this solution, and incorporate CSI feedback constraints.
In order to overcome all the aforementioned issues,
we introduce the modified QoS metric in \eqref{eqn:QoS2constraint}.
The advantages of using this metric are the following:
\begin{itemize}
\item It enables the development of low complexity beamforming solution to the stochastic optimization problem. 
It also enables us to analytically characterize the gap between the time average power expenditure of this algorithm and $P_{\text{inf}}$ (recall that the 
solution obtained by the Lyapunov optimization technique is sub-optimal).
\item Furthermore, the beamforming design algorithm is decentralized in which only the knowledge of local CSI is required, and 
the SCBSs only have to exchange scalar variables.
\item It enables the
development of a low complexity joint CSI feedback and beamforming design. Furthermore, this algorithm is also decentralized. 
\end{itemize}

\subsection{Extension to the case of Multi-tier Network}
The system model considered in this work can be extended to different Hetnet scenarios in a straightforward manner.
To illustrate this point, note that instead of viewing the system as consisting of $N$ small cells, 
one can simply consider the system as a network consisting of $N$ tiers. 
The QoS constraints can now be stated as follows:
\begin{align}
\gamma_{i,j}[t] & = \nu^U_{i} |\wv^H_{i,j}[t]  \hv_{i,i,j}[t]|^2 \nonumber \\ 
&- \sum^N_{n = 1}\nu^I_{n,i,j} \LB \sum^K_{k = 1}  |\wv^H_{n,k}[t]\hv_{n,i,j}[t]|^2 \RB - N_0. 
\label{eqn:QoSConstraintNew}
\end{align}
In \eqref{eqn:QoSConstraintNew}, $|\wv^H_{i,j}[t]  \hv_{i,i,j}[t]|^2$ is the useful signal from the BS of the $i^{\text{th}}$ tier to UT$_j$ in its own tier,
and $\sum^K_{k = 1}  |\wv^H_{n,k}[t]\hv_{n,i,j}[t]|^2$ is the interference arising from the BS of the $n^{\text{th}}$ tier $(n \neq i)$  at the UT$_j$ 
in the $i^{\text{th}}$ tier. 
$\nu^U_{i} $  and $\nu^I_{n,i,j}$ can be interpreted as constants that scale the useful signal, and the sum of the interference caused by the BS of the
$n^{\text{th}}$ tier on to a UT$_j$ in the $i^{\text{th}}$ tier. These constants can be set to suitable values in order to adapt 
the model for different Hetnet scenarios.
The time average QoS constraint can now be set as 
\begin{align}
\lim_{T \to \infty} \frac{1}{T} \sum^{T-1}_{t = 0} \mathbb{E} \LSB {\gamma}_{i,j}[t] \RSB  \geq \lambda_{i,j}.
\end{align}
{\bf Example:} Consider a $2$ tier network consisting of a macro-cell and a small-cell. 
Let us assume that the macro-cell is indexed by $i = 1,$ and the small cell by $i = 2.$
In this case, 
suppose the macro-cell users demand a certain average interference temperature constraint from the 
small cells. This can be accomplished by setting $\nu^U_{1} = 0$ and $\nu^I_{2,1} = -1.$
The QoS constraint becomes 
\begin{align}
\lim_{T \to \infty} \frac{1}{T} \sum^{T-1}_{t = 0} \mathbb{E} \LSB  \sum^K_{k = 1}  |\wv^H_{2,k}[t]\hv_{2,1,j}[t]|^2 \RSB \leq \lambda_{1}.
\end{align}
Therefore, we can extend the  system model of this paper to different 
HetNet scenarios.

We now proceed to the algorithm design. 
In the rest of the paper, we use the short hand notation 
$\bar{X}$ to denote the time average value of the random process 
$X[t],$ i.e., $\bar{X} = \lim_{T \to \infty} \frac{1}{T} \sum^{T-1}_{t =0}\mathbb{E}[X[t]].$ 

\section{Solution by Lyapunov Optimization Technique}
\label{sec:PerfectCSI}
In order to model the time average QoS constraint, we use the concept of 
 \emph{virtual queue} \cite{neelybook2010}. The virtual queue associated with
 the time average constraint $\bar{\gamma}_{i,j} \geq \lambda_{i,j}$ evolves in the following manner,
 \begin{align}
Q_{i,j} [t+1] = &  \max \LB Q_{i,j}[t] - \mu_{i,j}[t],0 \RB + A_{i,j}[t] \label{eqn:QLevolution_notations},
\end{align}
where $A_{i,j}[t]$ denotes the arrival process 
and $\mu_{i,j}[t]$ denotes the departure process.
The arrival and departure process are given by 
\begin{align}
A_{i,j}[t] & =  \nu_{i,j}\sum_{\substack{(n,k) \\ \neq (i,j)}}  |\wv^H_{n,k}[t]\hv_{n,i,j}[t]|^2+ \nu_{i,j} N_0 + \lambda_{i,j}\label{eqn:Adef} \\
\mu_{i,j}[t] & = |\wv^H_{i,j}[t]\hv_{i,i,j}[t]|^2. \label{eqn:mudef} 
\end{align}
Note that 
the arrival and the departure process can be upper bounded as
follows:
\begin{align}
A_{i,j}[t]  & \leq  NK \nu_{\max} P_{\text{peak}} H_{\max}+ \nu_{\max}N_0+\lambda_{i,j} \defines A_{\max} \label{eqn:UpBoundQoS21}, \\
\mu_{i,j}[t]  & \leq  P_{\text{peak}}H_{\max} \defines \mu_{\max} \label{eqn:UpBoundQoS22}.
\end{align}
The basic idea is to ensure stability of the virtual queue, stated mathematically as 
$\sum_{i,j} \bar{Q}_{i,j} < \infty.$
Ensuring the stability of the virtual queue implies that 
the time average of the arrival process is less than or equal to the service process, i.e.
$\bar{A}_{i,j} - \bar{\mu}_{i,j} \leq 0 \ \forall i,j.$
In other words, the constraint \eqref{eqn:const1} is satisfied. Thus, we reformulate the original problem into minimizing the time average power expenditure while stabilizing the virtual queue as follows:
\beqa
 &\dsp \min & \sum^N_{i = 1} \bar{P}_i \\
& s.t. &  \sum_{i,j} \bar{Q}_{i,j} < \infty, \ \
 \sum^K_{j = 1} \wv^H_{i,j}[t]\wv_{i,j}[t] \leq P_{\text{peak}}.  \nonumber
 \eeqa
With this reformulation, we solve the optimization problem 
in \eqref{eqn:OPT_basicEE} using the technique of
Lyapunov optimization \cite{neelybook2006}, which allows
us to consider the joint problem of stabilizing the queue and performance optimization. 
To this end, we define the quadratic Lyapunov function $\Psi : \RR^N \to \RR$ as:
$\Psi(\Qm[t]) = \frac{1}{2}\sum_{i,j} (Q_{i,j}[t])^2.$
The Lyapunov function is a scalar measure of the aggregate queue-lengths in the system.
We define the one-step conditional Lyapunov drift as 
\begin{align}
\Delta(\Qm[t]) = \mathbb{E} \LSB \Psi(\Qm[t+1]))- \Psi(\Qm[t]) \Big{|} \Qm[t]\RSB,
\end{align}
where the expectation is with respect to the random channel states
and the (possibly random) control actions made in reaction to these channel states. 

We now examine the Lyapunov drift corresponding to the evolution of the virtual queue $Q_{i,j}$. 
\begin{proposition}
\label{prop:QLBoundm}
For the virtual queue which evolves according to \eqref{eqn:QLevolution_notations}, the Lyapunov drift follows the condition
\begin{align}
& \Delta  (\Qm[t]) \leq C_1    - \sum_{i,j} \mathbb{E} \Big{[} Q_{i,j}[t] (A_{i,j}[t]-\mu_{i,j}[t]) \Big{|} \Qm[t] \Big{]} \ \forall t, \label{eqn:papa501m}
\end{align}
where 
\begin{align}
C_1 = NK \big{(}(A_{\max})^2 + (\mu_{\max})^2 \big{)} < \infty ,
\label{eqn:C1defm} 
\end{align}
and 
$A_{\max}$ and $\mu_{\max}$ represent the upper bound on
the arrival and departure process specified in \eqref{eqn:UpBoundQoS21} and \eqref{eqn:UpBoundQoS22} respectively.
\end{proposition}
The proof is provided in Appendix A, part I.

Adding the cost term (i.e. transmit power during the current time slot),
$V \mathbb{E} \big{[}  \sum_{i,j} \wv^H_{i,j} [t]\wv_{i,j}[t] \big{|} \Qm[t] \big{]}$\footnote{The conditional expectation of the transmit power with respect to $\Qm[t]$ is taken since the optimal transmit power during every time slot will be a function of the virtual queue-length values.} to both the sides of the equation
\eqref{eqn:papa501m}, 
we obtain
\begin{align}
& \Delta (\Qm[t]) + V \mathbb{E} \big{[}  \sum_{i,j} \wv^H_{i,j} [t]\wv_{i,j}[t] \big{|} \Qm[t] \big{]} \leq C_1 \nonumber \\ & - \sum_{i,j} \mathbb{E} \Big{[} Q_{i,j}[t] ( A_{i,j}[t] - \mu_{i,j}[t] )  -V  \wv^H_{i,j} [t]\wv_{i,j}[t]\Big{|} \Qm[t]\Big{]} \ \forall t. \label{eqn:papa501}
\end{align}
Henceforth, we call the term 
$\Delta_V  (\Qm[t]) \defines \Delta  (\Qm[t]) + V \mathbb{E} \big{[}  \sum_{i,j} \wv^H_{i,j} [t]\wv_{i,j}[t] \big{|} \Qm[t] \big{]}$ as the modified Lyapunov drift.
The bound on the modified Lyapunov drift for the 
queue-length evolution can be given as
\begin{align}
& \Delta(\Qm[t])  \leq C_1 +   \sum_{i,j}Q_{i,j}[t] (\lambda_{i,j}+\nu_{i,j} N_0) \nonumber \\ &   - \sum_{i,j} \mathbb{E} \Big{[} Q_{i,j}[t] \big{(} |\wv^H_{i,j} [t]\hv_{i,i,j} [t]|^2 \nonumber \\ & - \nu_{i,j} \sum_{\substack{(n,k) \\ \neq (i,j)}}|\wv^H_{n,k} [t] \hv_{n,i,j} [t]|^2 \Big{)}  -V  \wv^H_{i,j} [t]\wv_{i,j}[t]\Big{|} \Qm[t]\big{]} \ \forall t, \label{eqn:papa51}
\end{align}
where we have used the expressions for $A_{i,j}[t]$
and $\mu_{i,j}[t]$ from \eqref{eqn:Adef} and \eqref{eqn:mudef} respectively.

Following the approach of  Lyapunov optimization  (drift-plus penalty method), 
we design the beamforming vectors during each 
time slot $t$ to minimize the bound on the Lyapunov
drift during each time slot \cite{neelybook2006}. 
Before we proceed, we state the main intuition of using the Lyapunov optimization method
in the design of the algorithm.

The modified Lyapunov drift has two components, the Lyapunov
drift term $\Delta(\Qm[t]),$ and the penalty term $V \times \mathbb{E}   \big{[} \sum_{i,j} \wv^H_{i,j} [t]\wv_{i,j}[t]  \big{|} \Qm[t] \big{]}$ term. Intuitively, minimizing the Lyapunov drift term 
alone pushes the queue-length of the virtual energy queue to 
a lower value. 
The second metric $V \times \mathbb{E}   \big{[} \sum_{i,j} \wv^H_{i,j} [t]\wv_{i,j}[t]   \big{|} \Qm[t] \big{]}$ can be viewed 
as the penalty term for using high downlink power, with the parameter $V$ representing 
the trade-off between minimizing the queue-length drift
and minimizing the penalty function. Greater
value of $V$ represents greater priority 
to minimizing the downlink power at the expense of
greater size of the virtual energy queue and vice versa.
Therefore in our algorithm, we minimize the modified Lyapunov drift instead of only minimizing $\Delta(\Qm[t]) ,$
and obtain a trade-off between the two metrics of interest.
Furthermore, we theoretically examine some properties of this algorithm and analyze its
performance. In particular, we characterize the sub optimality in the performance of this algorithm with respect
to $P_{\text{inf}}$, i.e. the minimum power expenditure over all feasible control policies

Accordingly, the beamforming vector should be computed
as a solution to the following optimization problem (minimize the bound obtained in \eqref{eqn:papa51}):
\begin{align}
& \wv_{i,j}[t]  \in \argmax_{\wv_{i,j}, \forall \{i,j\} }  \sum_{i,j} \mathbb{E}_{\Hm} \Big{[}Q_{i,j}[t] |\wv^H_{i,j}\hv_{i,i,j}[t]|^2  \nonumber \\ & - \nu_{i,j} Q_{i,j} [t]   \sum_{\substack{(n,k) \\ \neq (i,j)}}|\wv^H_{n,k} \hv_{n,i,j}[t]|^2 - V  \wv^H_{i,j}\wv_{i,j} \Big{|} \Qm[t]\Big{]}.  \label{eqn:OPT_multcell_Lyp}
\end{align}
where $\mathbb{E}_{\Hm}$ indicates that the expectation is
with respect to the random channel realization.
%
\subsection{Algorithm Design with Perfect Knowledge of Local CSI at the BS}
In this section, we assume that the SCBSs have the perfect knowledge of CSI of all its downlink channels ($\hv_{i,n,k} \ \forall n,k$)  and examine the optimization problems \eqref{eqn:OPT_multcell_Lyp}. The CSI information can be obtained 
by feedback from the UTs\footnote{A more practical case where the SCBS can obtain CSI feedback from only a limited number of UTs will be addressed in Section \ref{sec:FB_BF_Opt}.}.
With the perfect knowledge of CSI, \eqref{eqn:OPT_multcell_Lyp}
reduces to greedily minimizing the term inside the expectation ($\mathbb{E} [f(Y)|Y] = f(Y)$).
Therefore, we remove the expectation and solve the following optimization problem (we drop the time index $t$).
\beqa 
&\dsp \max_{\wv_{i,j}, \forall \{i,j\} }& \sum_{i,j} \Big{[} Q_{i,j} |\wv^H_{i,j}\hv_{i,i,j}|^2 \label{eqn:OPT_multcell_Lyp_notime} \\ && - \nu_{i,j} Q_{i,j}    \sum_{\substack{(n,k) \\ \neq (i,j)}}|\wv^H_{n,k} \hv_{n,i,j}|^2 - V  \wv^H_{i,j}\wv_{i,j} \Big{]} \nonumber \\
& s.t. & \sum_{j} \wv_{i,j}^H  \wv_{i,j} \leq P_{\text{peak}}  \qquad \forall i. \nonumber
 \eeqa

We now examine \eqref{eqn:OPT_multcell_Lyp_notime} in
greater detail.
The objective function of the optimization problem in \eqref{eqn:OPT_multcell_Lyp_notime} can be rearranged and written as,
 \beqa 
&\dsp \max_{\wv_{i,j}, \forall \{i,j\} } & \sum_{i,j} \wv^H_{i,j} \Am_{i,j} \wv_{i,j} \label{eqn:opt_quadform}\\
& s.t. & \sum_{j} \wv_{i,j}^H  \wv_{i,j}\leq P_{\text{peak}}  \qquad \forall i \nonumber
 \eeqa
 where the matrix 
 $ \Am_{i,j} = Q_{i,j} \Hm_{i,i,j}-  \sum_{\substack{(n,k) \\ \neq (i,j)}} \nu_{n,k} Q_{n,k} \Hm_{i,n,k}  - V \Id$
 and $\Hm_{i,n,k} = \hv_{i,n,k} \hv^H_{i,n,k}.$
  Note that the optimization problem in \eqref{eqn:opt_quadform}
 is in separable form, in which
 each SCBS $i$ can solve the sub problem given by 
 \beqa 
&\dsp \max_{\wv_{i,j}, \forall \{i,j\} } & \sum_{j} \wv^H_{i,j} \Am_{i,j} \wv_{i,j} \label{eqn:opt_quadform_ind} \\
& s.t. & \sum_{j} \wv_{i,j}^H  \wv_{i,j}\leq P_{\text{peak}} \nonumber.  
 \eeqa
We now provide an algorithm to solve the 
optimization problem \eqref{eqn:opt_quadform_ind} and address it by the name 
Decentralized Beamforming algorithm - {\bf DBF}.
We also denote the beamforming vector corresponding to the DBF algorithm by $\wv^{\text{DBF}}_{i,j}.$
Also, let us we denote $\lambda^{\max} (\Am_{i,j})$ as the maximum eigenvalue of the matrix $\Am_{i,j}.$
\vspace{0.05in} \hrule
\vspace{0.01in}\hrule\vspace{0.05in}
\begin{algorithm}[Decentralized Beamforming algorithm - {\bf DBF}]
During each time slot $t,$ perform the following steps:
\begin{itemize}
\item Compute $ j^* = \argmax_j \lambda^{\max}(\Am_{i,j}).$
\item Set the beamforming vector corresponding to the UT $j^*$ as follows:
\begin{align}
\wv^{\text{DBF}}_{i,j^*} = \LB P_{\text{peak}}\lambda^{\max}(\Am_{i,j^*}) {\bf 1}_{\lambda^{\max}(\Am_{i,j^*}) > 0} \RB \xv_{\lambda^{\max}(\Am_{i,j^*})} \label{eqn:papa112}
\end{align}
where ${\bf 1}_{\lambda^{\max}(\Am_{i,j^*}) > 0}$ represents the indicator function whose value is $1$ if
$\lambda^{\max}(\Am_{i,j^*}) > 0,$ and $0$ otherwise, and $\xv_{\lambda^{\max}(\Am_{i,j^*})}$ is the eigenvector corresponding to the 
maximum eigenvalue of matrix $\Am_{i,j^*}.$ 
\item For all other UTs $j (\neq j^*),$ set
\begin{align}
\wv^{\text{DBF}}_{i,j} = {\bf 0} \qquad j \neq j^*.
\end{align}
\end{itemize}
\end{algorithm}
\hrule 
\vspace{0.01in} \hrule \vspace{0.05in}
From the steps of the DBF algorithm, it can be verified that the optimal direction
of transmission to solve \eqref{eqn:opt_quadform_ind} is to transmit
along the eigenvector corresponding to 
the maximum eigenvalue of the matrix $\Am_{i,j^*},$

For simplicity of notations, we denote 
$\Wm_{i,j}[t] = \wv_{i,j}[t]\wv^H_{i,j}[t]$
and hence
$\wv^H_{i,j} \Am_{i,j} \wv_{i,j} = \trace(\Am_{i,j}[t]\Wm_{i,j}[t]).$

The solution implies that at most one UT can be active per cell during each time slot. 
Also, we can conclude that
\begin{align*}
\sum_{j} &\trace(\Am_{i,j}[t]\Wm^{\text{DBF}}_{i,j}[t]) = P_{\text{peak}}\lambda^{\max}(\Am_{i,j^*}) {\bf 1}_{\lambda^{\max}(\Am_{i,j^*}) > 0}.  \end{align*}

\subsection{Properties of the DBF algorithm}
We now provide some properties of the DBF algorithm.
$\bullet$ {\bf Intuition:} 
Taking a closer look at the optimization problem \eqref{eqn:opt_quadform_ind}, it can be seen that each UT$_{i,j}$
has a metric associated with it given by,
\begin{align}
\trace & (\Am_{i,j}\Wm_{i,j})  = \trace \big{(}  Q_{i,j} \underbrace{\Hm_{i,i,j} \Wm_{i,j}}_{\text{Useful signal}} \nonumber \\ & -  \sum_{\substack{(n,k) \\ \neq (i,j)}} \nu_{n,k} Q_{n,k} \underbrace{\Hm_{i,n,k} \Wm_{i,j}}_{\text{Interference signal to other UTs}} - V \Wm_{i,j}  \big{)}. \label{eqn:Intuition}
\end{align}
The metric corresponds to the difference between weighted sum of the useful signal (to the UT$_{i,j}$) and the weighted sum of
interference caused to the other UTs in the system (UT$_{{(n,k) \\ \neq (i,j)}}).$
The weights are the corresponding queue-length values which indicate how urgently the UT needs to be served.
Therefore, intuitively, each SCBS schedules the UT in its cell which has the highest value 
of this metric $\lambda^{\max} (\Am_{i,j^*}).$
Additionally, the transmission direction corresponds to the eigenvector corresponding to the $(\lambda^{\max} ( \Am_{i,j^*})).$
The parameter $V$ represents how aggressively the SCBS decides to transmit. Higher value of $V$ implies less aggressive 
transmission and greater energy savings. 

The beamforming vector that maximizes the metric in \eqref{eqn:Intuition} 
has some similarities with
the ``Leakage-based beamforming design" \cite{SayedLeakage2007}, \cite{DingLeakageTSP2013}, \cite{DingLeakageTWC2014},
where in the beamforming vectors are chosen to maximize the SLNR metric. 
Since the original QoS metric considered in this work is in the difference form (as opposed to the ratio form of
the SLNR metric), our solution chooses the beamforming vector that maximizes the
weighted difference between the useful signal and the interference caused to the other UTs in the system (and not the ratio as in the case of SLNR).
Moreover, the useful signal power and the interference signal powers are scaled by the respective virtual queue lengths during that time slot. 
Therefore, our algorithm can be viewed as a ``dynamic time varying leakage based algorithm" where the impacts of useful signal and interference 
signal are adapted dynamically according to the achieved QoS (represented by the virtual queue length levels). 
\\
$\bullet$ {\bf Decentralized Solution:} Observe that in order to formulate the matrix $\Am_{i,j},$ the SCBSs
only require the local CSI ($\hv_{i,n,k} \ \forall n,k$). The SCBSs would only have to exchange 
the queue-lengths ($Q_{i,j}$) among themselves.
Therefore, our formulation naturally leads to a decentralized solution. This results in tremendous reduction in the backhaul 
capacity requirements. 

\subsection{Performance bounds for the DBF algorithm:} 
The following proposition provides the performance bounds 
for the DBF algorithm:
\begin{proposition}
\label{prop:DBFperformance}
The DBF algorithm yields the following performance bounds.
The virtual queue is strongly stable and for any $\epsilon> 0$,$V \geq 0,$ the time average queue-length satisfies
\begin{align}
\sum_{i,j}  \bar{Q}^{\text{DBF}}_{i,j} \leq \frac{C_1+V NK P_{\text{peak}}}{\epsilon}
\label{eqn:DBFRes1}
\end{align}
and the time average energy expenditure yields, 
\begin{align}
\sum_{i} \bar{P}^{\text{DBF}}_i
\leq  P_{\inf}+\frac{C_1}{V} \label{eqn:DBFRes2}.
\end{align}
\end{proposition}
\begin{proof}
The proof is provided in Appendix B, part II.
\end{proof}
The bound of Proposition \ref{prop:DBFperformance} implies that the time average energy expended by the DBF 
algorithm can be made arbitrarily close to the minimum average power (over all possible sequence on control actions)
by increasing the value of $V$ to an arbitrarily high value. 
This comes at the expense of increasing the average queue-length of 
the virtual queue. 
Intuitively, a high value of the average queue-length implies 
that the number of time slots required to satisfy the time average 
constraints is higher (analogous to the concept of delay in real queues).

\subsection{A Note on using the time average SINR for algorithm design}
Similar to the derivation at the beginning of \eqref{eqn:papa51}, it can be shown that the use of time average SINR
(as the QoS metric) in the Lyapunov optimization method leads to
solving the following instantaneous optimization
problem during each time slot $t:$
\beqa 
&\dsp \max_{\wv \in \CC^{N_t}} & \sum_{i,j} \Big{[} Q_{i,j}[t] \text{SINR}_{i,j} - V  \wv^H_{i,j}\wv_{i,j}  \Big{]} \label{eqn:OPT_multcell_Lyp_notime_SINR}   \\
& s.t. & \sum_{j} \wv_{i,j}^H  \wv_{i,j} \leq P_{\text{peak}}  \qquad \forall i. \nonumber
 \eeqa
Notice that the optimization
problem in \eqref{eqn:OPT_multcell_Lyp_notime_SINR}
corresponds to solving a maximization of the weighted sum of 
SINR terms, which is a non-convex problem,
and finding the global optimum is a non-trivial task.
Additionally, it is very difficult to obtain a decentralized
formulation of the solution corresponding to
\eqref{eqn:OPT_multcell_Lyp_notime_SINR}, and to develop an efficient CSI feedback strategy in the case when there are CSI feedback constraints.
This highlights the functional importance of using the
QoS metric in \eqref{eqn:QoS2constraint} in our algorithm design.

Next, we introduce delays in the information 
exchange among the SCBSs.

\section{Delayed Queue-length Information Exchange}
\label{sec:DelInfoExchg}
In this section, we make a novel contribution by studying the impact of delayed information exchange between SCBSs (such as the delays introduced in the backhaul).
Recall that in our beamforming solution, the SCBSs can formulate
the beamforming vectors using only the local CSI.
However, the SCBSs have to exchange the queue length information (scalar values). In this section, we show that this exchange does not have to be in real time, and our solution can be made arbitrary close to infimum even if the exhcange is done with delay. We study the proposed solution in presence of delayed information exchange between SCBSs, and show the impact of this delay on the gap between our solution, and the minimum transmit power of the original stochastic problem.

\subsection{Algorithm Design with Delayed Information Exchange between the BSs}
Let us assume that a delay of $\tau < \infty$ time slots is incurred 
while the SCBSs exchange the queue-length information.
Each SCBS $i$ now has perfect queue-length information of its local queues ($Q_{i,j}[t] \ \forall j$) and the delayed queue-length
information from the neighboring queues ($Q_{n,k}[t-\tau], \ \forall n \neq i,k$).
Note that our set up can be easily generalized to introduce different delays $\tau_n, \forall n \neq i$ 
corresponding to the queue-length information from different SCBSs. However, in order to keep the notations simple, we restrict ourselves to uniform delays ($\tau, \ \forall n \neq i$).
We assume that the SCBSs treat the delayed queue-length 
as the true value of the queue-length.
Every SCBS now solves the following optimization problem,
\beqa 
&\dsp \max_{\wv} & \sum_{j}  \trace \LB \Am^\tau_{i,j} [t] \Wm_{i,j} \RB  \label{eqn:OPT_multcellLyp_delayed} \\
& s.t. & \sum_j \trace(\Wm_{i,j}) \leq P_{\text{peak}}  \nonumber
 \eeqa
 where the matrix $\Am^\tau_{i,j}[t]$ is given by,
 \begin{align}
\Am^\tau_{i,j}[t] & = Q_{i,j}[t] \Hm_{i,i,j}(t)- \sum_{\substack{k \neq j}} \nu_{i,k} Q_{i,k}[t] \Hm_{i,i,k}[t]  \nonumber
\\& -  \sum_{\substack{n \neq i, k}} \nu_{n,k} Q_{n,k}[t-\tau] \Hm_{i,n,k}[t]  - V \Id.
\end{align}  
Let us denote the solution corresponding to optimization problem \eqref{eqn:OPT_multcellLyp_delayed}
by $\Wm^{\text{del}}[t].$ We will henceforth use the superscript "del" to
denote parameters corresponding to the solution of \eqref{eqn:OPT_multcellLyp_delayed}.
Once again, following similar argument as 
the solution to \eqref{eqn:opt_quadform_ind},
it can be shown that at most one UT can be active per cell.
\vspace{0.05in} \hrule
\vspace{0.01in}\hrule\vspace{0.05in}
\begin{algorithm}[DBF algorithm with delayed information exchange]
During each time slot $t,$ perform the following steps:
\begin{itemize}
\item Compute $j^*_\tau = \argmax_{j} \lambda^{\max}(\Am^\tau_{i,j}).$
\item Set the beamforming vector corresponding to the UT $j^*$ as follows:
\begin{align}
\wv^{\text{del}}_{i,j^*_\tau } = \LB P_{\text{peak}} \lambda^{\max}(\Am_{i,j^*_\tau}) {\bf 1}_{\lambda^{\max}(\Am_{i,j^*_\tau}) > 0} \RB \xv_{\lambda^{\max}(\Am_{i,j^*_\tau})} \label{eqn:papa1112}
\end{align}
where ${\bf 1}_{\lambda^{\max}(\Am_{i,j^*_\tau}) > 0}$ represents the indicator function whose value is $1$ if
$\lambda^{\max}(\Am_{i,j^*_\tau}) > 0,$ and $0$ otherwise, and $\xv_{\lambda^{\max}(\Am_{i,j^*_\tau})}$ is the eigenvector corresponding to the 
maximum eigenvalue of matrix $\Am_{i,j^*_\tau}.$ 
\item For all other UTs $j (\neq j^*_\tau),$ set
\begin{align}
\wv^{\text{del}}_{i,j} = {\bf 0} \qquad j \neq j^*_\tau.
\end{align}
\end{itemize}
\end{algorithm}
\hrule 
\vspace{0.01in} \hrule \vspace{0.05in}

It is clear that the optimal 
power allocation policy is given by 
\begin{equation}
P^{\text{del}}_{i,j} = \begin{cases} P_{\text{peak}}  & \ \text{if} \ j = j^*_\tau \text{and} \ \lambda^{\max}(\Am_{i,j^*_\tau}) > 0  \\
                      0  & \text{else},                     
                   \end{cases} 
\end{equation}
and therefore,
\begin{align}
\max_{\wv_{i,j}} \sum_{j} & \trace(\Am^\tau_{i,j}[t]\Wm^{\text{del}}_{i,j}[t])   = P_{\text{peak}}\lambda^{\max}(\Am_{i,j^*_\tau}) {\bf 1}_{\lambda^{\max}(\Am_{i,j^*_\tau}) > 0}.  \label{eqn:papa8EE2}
\end{align}
We ow theoretically examine the performance of the DBF algorithm with delayed queue-length information.

\subsection{Performance Analysis with Delayed Information Exchange Between the BSs}
We first compare the performance of DBF algorithm with respect to the solution of \eqref{eqn:OPT_multcellLyp_delayed} (i.e. the case with delayed queue-length exchange) in the following lemma.
\begin{lemma}
\label{lem:QLDelBd}
There exists a constant $0 \leq C_2 < \infty$ independent of the current queue-length $Q_{i,j}[t], \ \forall i,j$ such that,
\begin{align}
\sum_{i,j}\trace(\Am_{i,j}[t]\Wm^{\text{DBF}}_{i,j}[t]) \leq  \sum_{i,j}\trace(\Am_{i,j}[t]\Wm^{\text{del}}_{i,j}[t])+C_2 \ \forall t. \label{eqn:lemma_del}
\end{align}
\end{lemma}
\begin{proof}
The lemma is proved in Appendix B, part I.
\end{proof}
Lemma \ref{lem:QLDelBd} states that the performance of the DBF algorithm 
with delayed queue-length information exchange
differs from that of DBF algorithm (with instantaneous queue-length
information exchange) by a bounded constant.
The key element in this lemma is the fact that the constant $C_2$ is independent of the current queue-lengths
which will be helpful
in proving the performance bounds for the
DBF algorithm with delayed queue-length information
exchange.
\begin{corollary}
\label{thm:DelQL}
For the DBF algorithm with delayed queue-length information exchange, the following
performance can be obtained for any $\epsilon >0$ and $V \geq 0$: 
The time average transmit power expenditure satisfies
\begin{align}
 \sum^N_{i = 1} \bar{P}^{\text{del}}_i
\leq  P_{\inf}+\frac{C_1+C_2}{V}. \label{eqn:PowBd_Del}
\end{align}
and the time average queue-length satisfies
\begin{align}
 \sum_{i,j} \bar{Q}^{\text{del}}_{i,j} \leq \frac{C_1+C_2+V NK P_{\text{peak}}}{\epsilon} \label{eqn:QLBd_Del}
\end{align}
where $C_1$ is defined in \eqref{eqn:C1defm} and $C_2$ defined from Lemma \ref{lem:QLDelBd}.
\end{corollary}
Corollary \ref{thm:DelQL} is proved in Appendix B, part II.

The corollary shows that by increasing the value of $V,$
the time average power expenditure can be made very close 
to the optimal value of the original problem, i.e., the impact
of the delays in information exchange can be made negligible.
However, this will have an impact on the time needed 
to achieve the time average QoS constraint (i.e., we need 
greater number of time slots to achieve the QoS).
This result shows the SCBSs do not have to exchange the queue-length information in real time. The SCBSs can delay this exchange, and even adapt it to the backhaul capacity/load. Under some cases, the impact of the delay in information exchange can be even made negligible (for e.g., choosing a high value of the parameter $V$).

\section{Joint optimization
of the CSI Feedback and Transmission}
\label{sec:FB_BF_Opt}
In this Section, we solve the problem of joint CSI feedback and beamforming design by using the Lyapunov optimization framework. The solution itself is not straightforward from the Lyapunov optimization technique, and requires several intermediate proofs that are presented in Theorem \ref{thm:QLFedBd}.

\subsection{Algorithm Design : Joint CSI Feedback and Transmission}
{\bf CSI Feedback Model :} In practice, the UTs feedback a quantized
version of the CSI to the SCBS.
However, in this work, owing the complexity of the 
joint feedback and beamforming design problem, 
we consider a simple feedback scheme.
Under this scheme, we assume that 
if SCBS$_i$ decides to obtain the CSI feedback from the UT$_{n,k},$
then the UT can feedback this information perfectly. 
Hence, SCBS$_i$ has the knowledge of the exact value of $\hv_{i,n,k}.$
For the rest of the UTs (that do not feedback their CSI), SCBS$_i$ assumes the channel to be the mean value, i.e. $\mathbb{E}[\hv_{i,n,k}].$ Consequently, the SCBS must decide which of the
UTs have to feedback their CSI\footnote{The impact of channel 
quantization errors will be addressed in the future work}.
We impose the following feedback constraint: during every time slot $t,$ the SCBS can obtain feedback from at-most $B_{\max}$
UTs (recall that in practice, the SCBS can obtain the feedback
from only a subset of the UTs). 
We denote the indicator variable for the CSI feedback
decision by
 \begin{equation}
b_{i,n,k} = \begin{cases} 1  & \ \text{if} \ \text{UT$_{n,k}$ feeds back its CSI to SCBS$_i$}  \\
                      0  & \text{else}.                     
                   \end{cases} \nonumber
\end{equation}
The feedback constraint can then be stated as 
\begin{align}
\sum_{n,k} b_{i,n,k} \leq B_{\max}. \label{eqn:FBscheme2}
\end{align}
 Let us denote
 $\bv_{i,j} = [b_{i,j,1},\dots,b_{i,j,K}];$
 and $\bv_{i} = [\bv_{i,1},\dots,\bv_{i,N}]$
 and $\bv = [\bv_{1},\dots,\bv_{N}].$
Also, for simplicity of notations, we henceforth denote
\begin{align*}
\tilde{\Hm}_{i,n,k} & = \Hm_{i,n,k} - \mathbb{E} [ \Hm_{i,n,k}] \\
\bar{\Hm}_{i,n,k} & = \mathbb{E} [ \Hm_{i,n,k}] \\
\tilde{Q}_{i,n,k} & = {Q}_{i,n,k} \sigma_{i,n,k}.
\end{align*}
The joint CSI feedback decision and transmission algorithm can be derived by following steps similar to
Proposition \ref{prop:QLBoundm} and the beamforming vector
design of \eqref{eqn:OPT_multcell_Lyp_notime} (omitted here for brevity). 
In what follows, we first present 
the joint CSI feedback decision and transmission rule obtained from the Lyapunov optimization technique (for the case with partial CSI), 
and then present the results on the performance analysis.

\vspace{0.05in} \hrule
\vspace{0.01in}\hrule\vspace{0.05in}
\begin{algorithm}[Joint CSI Feedback and Transmission Algorithm]
During each time slot $t,$ perform the following steps:
\begin{itemize}
\item Choose the CSI feedback decision as a solution to the following optimization problem:
\begin{align}
& \bv^{\text{FB}}  = \nonumber \\
& \max_{\bv}  \mathbb{E} \Big{[} \max_{\Pm} \sum_{i,j} P_{i,j}  \lambda^{\max}\Big{(} Q_{i,j} ( b_{i,i,j}\tilde{\Hm}_{i,i,j} + \bar{\Hm}_{i,i,j}) \nonumber  \\ & \qquad - \sum_{\substack{(n,k) \\ \neq (i,j)}} \nu_{n,k} Q_{n,k} {(}  b_{i,n,k} \tilde{\Hm}_{i,n,k} + \bar{\Hm}_{i,n,k}) - V \Id \Big{)} \Big{]},  \label{eqn:FBcostfn}
\end{align}
where we denote $\bv^{\text{FB}}$ to be the resulting
CSI feedback decision.
\item After computing $\bv^{\text{FB}},$ 
if $b^{\text{FB}}_{i,n,k} = 1,$ set the channel vector from SCBS$_i$
to UT$_{n,k}$ to $\hv_{i,n,k}$ (i.e., the actual value of the channel realization).
Else, set the channel vector from SCBS$_i$
to UT$_{n,k}$ to $\mathbb{E}[\hv_{i,n,k}].$ After setting the corresponding channel values,
compute the beamforming vectors by repeating the steps
of Algorithm 1.
\end{itemize}
\end{algorithm}
\hrule 
\vspace{0.01in} \hrule \vspace{0.05in}

\subsection*{A Low Complexity Algorithm to Solve \eqref{eqn:FBcostfn}}
Notice that the feedback design problem in \eqref{eqn:FBcostfn} is a stochastic optimization problem with discrete variables $\bv,$ and over the power allocation $\Pm$. 
We propose the following 
low complexity implementation
scheme, labeled $AlgF$ to
optimally solve \eqref{eqn:FBcostfn}.
Before we do so, let us define
\begin{align}
& \tilde{\Qm}_{ii} = \{\tilde{Q}_{i,1},\tilde{Q}_{i,2},\dots,\tilde{Q}_{i,K} \} \\
& \tilde{\Qm}_{ii_-} = \{\nu_{n,1}\tilde{Q}_{n,1},\nu_{n,2}\tilde{Q}_{n,2},\dots,\nu_{n,K}\tilde{Q}_{n,K} \} \ \ \forall n\neq i.
\end{align}
$AlgF$ can be implemented in the following manner:

\vspace{0.05in} \hrule
\vspace{0.01in}\hrule\vspace{0.05in}
\begin{algorithm}
[$AlgF$]
Perform the following steps.
\begin{itemize}
\item Initialize $b_{i,n,k} = 0,  \forall n,k.$ 
\item Choose the index $j^\prime = \argmax_{j} \tilde{\Qm}_{ii}$ and 
$k^\prime = \argmax_{k,n\neq i} \tilde{\Qm}_{ii_-}.$ Then, select either
the index $j^\prime$ or $k^\prime$ that maximizes the cost function in \eqref{eqn:FBcostfn}. If $j^\prime$ 
is selected, then set $b_{i,i,j^\prime} = 1$ or if $k^\prime$
is selected, then set $b_{i,n,k^\prime} = 1.$
\item If $j^\prime$ 
is selected, then update $\tilde{\Qm}_{ii}$ as $\tilde{\Qm}_{ii} = \tilde{\Qm}_{ii} \setminus \tilde{Q}_{i,j^\prime}$  
or if $k^\prime$ 
is selected, then update $\tilde{\Qm}_{ii_-}$ as $\tilde{\Qm}_{ii_-} = \tilde{\Qm}_{ii_-} \setminus \tilde{Q}_{n,k^\prime}.$  
If $\sum_{n,k} b_{i,n,k} \leq B_{\max},$ go to step 2, else terminate.
\end{itemize}
\end{algorithm}
\hrule 
\vspace{0.01in} \hrule \vspace{0.05in}
The aforementioned algorithm requires the computation of the max eigenvalues for $2 B_{\max}$ UTs. This results in a huge complexity reduction as compared to the exhaustive  search method that requires the computation of the max eigenvalues for $\binom{NK}{B_{\max}}$ combinations. 

It is worth noting that, in practice, the UTs in other cells i.e. $(n\neq i,k)$ may not be allowed to feedback their CSI to base station $i$. Under this condition, a similar feedback problem can be reformulated. The constraint \eqref{eqn:FBscheme2} is now over the UTs of cell $i$.  Our feedback algorithm can be modified taking this constraint into account, which results in a simple algorithm that consists of selecting the $B_{\max}$ UTs that have the highest $\tilde{Q}_{i,j}$ (the computation of the maximum eigenvalue is not needed in this case). 

\subsection{Algorithm Performance Analysis}
In this subsection, we provide the performance analysis of both $AlgF$ and the joint CSI feedback decision and transmission algorithm.
\begin{theorem}
\label{thm:QLFedBd}
The proposed feedback algorithm $AlgF$ is the optimal solution to the problem in \eqref{eqn:FBcostfn}.
\end{theorem}
The proof is given in appendix C, Part I. 

Let $P^\prime_{\text{inf}}$ be the minimum time average
transmit power $\bar{P}$, over all possible sequences of control actions of the optimization problem \eqref{eqn:OPT_basicEE}, with an additional constraint on the CSI feedback as in \eqref{eqn:FBscheme2}.
\begin{corollary}
The proposed joint feedback and beamforming algorithm can achieve the following performance : The average energy energy expenditure, and the time average virtual queue length satisfy the bounds given by
\begin{align}
 \sum_{i}  \bar{P}^{\text{FB}}_i 
 & \leq P^\prime_{\inf}+\frac{C_1}{V} \nonumber \\
 \sum_{i,j} \bar{Q}^{\text{FB}}_{i,j}  & \leq \frac{C_1+V NK P_{\text{peak}}}{\epsilon}
\end{align}
\label{corol:QLFedBd}
\end{corollary}
The proof  is provided in Appendix C, Part II. 
Recall that the time average QoS constraint is satisfied if the  time average queue lengths are bounded, which is ensured by the use our algorithm for finite values of $V$, $N$ and $K$.

\section{Numerical Results}
\label{sec:NumResults}
In this section, we present some numerical results to demonstrate the performance 
of the DBF algorithm.
We consider a system consisting of $2$ small cells with 
each cell having $2$ UTs each. Each SCBS has $5$ antennas and $P_{\text{peak}} = 10$dB per SCBS.
We consider a distance dependent path loss model, the path loss factor from SCBS$_i$ to UT$_{j,k}$ is given as
$\sigma_{i,j,k} = {d^{-\beta}_{i,j,k}}$
where $d_{i,j,k}$ is the distance between SCBS$_i$ to UT$_{j,k}$, normalized
to the maximum distance within a cell, and $\beta$ is the path
loss exponent (in the range from $2$ to $5$ dependent on the radio environment).
Once again, as stated before, we assume $N_0 = 1,$ and normalize
the signal powers by the noise variance.

We plot the time average energy expenditure per SCBS versus the target QoS for two cases.
In the first case, we solve the problem of minimizing the instantaneous energy expenditure subject to instantaneous QoS target
constraints ( $\min_{\wv} \sum^N_{i = 1} \sum^K_{j = 1} \wv^H_{i,j}[t] \wv_{i,j}[t] \ \text{s.t.} \ \gamma_{i,j}[t] \geq \lambda_{i,j} \ \forall t$). We repeat this for $1000$ time slots. In the second scenario, we solve the problem of minimizing the time average energy expenditure subject to time average QoS constraints ($\bar{\gamma}_{i,j} \geq \lambda_{i,j}$).
We plot the result in Figure \ref{fig:feas_SINREE}. It can be seen that for the case with time average constraints, the energy expenditure is lower.
In particular, for a target QoS of $10$dB, energy minimization with time average QoS constraints under the Lyapunov optimization based approach 
provides upto $4$dB reduction in the energy expenditure as compared to the case with instantaneous QoS constraints (for $V = 800.$)
This is in accordance with our intuition that the time average QoS constraint provides greater flexibility 
in allocating resources over channel fading states.
\begin{figure}[htp]
\begin{scriptsize}
 \includegraphics[width=3.5 in]{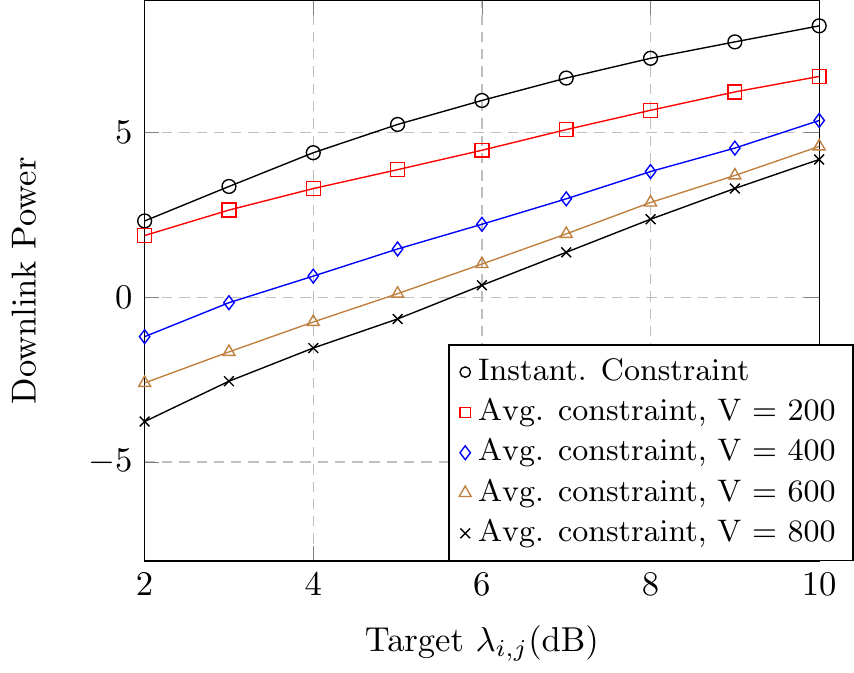}  
\caption{Average energy expenditure Vs target QoS for a two cell scenario, each cell consisting of two UTs, $N_t = 5,$ $P_{\text{peak}} = 10$dB.}
\label{fig:feas_SINREE}
\end{scriptsize}
 \end{figure} 

We also plot the achieved time average SINR as a function of the target QoS for different values of $V$ in Figure \ref{fig:feas_SINRVsTarget}.
It can be seen that in each of the cases, the achieved time average SINR in the downlink is above the target value $\lambda_{i,j}$. This result emphasizes the fact that although 
the QoS constraint of the form \eqref{eqn:QoS2constraint} was used in the algorithm design, the resulting algorithm still achieves the target value in terms of the achieved time average SINR. Thus, our algorithm can be directly used in the case with time average SINR constraints.
\begin{figure}[htp]
\begin{scriptsize}
 \includegraphics[width=3.5 in]{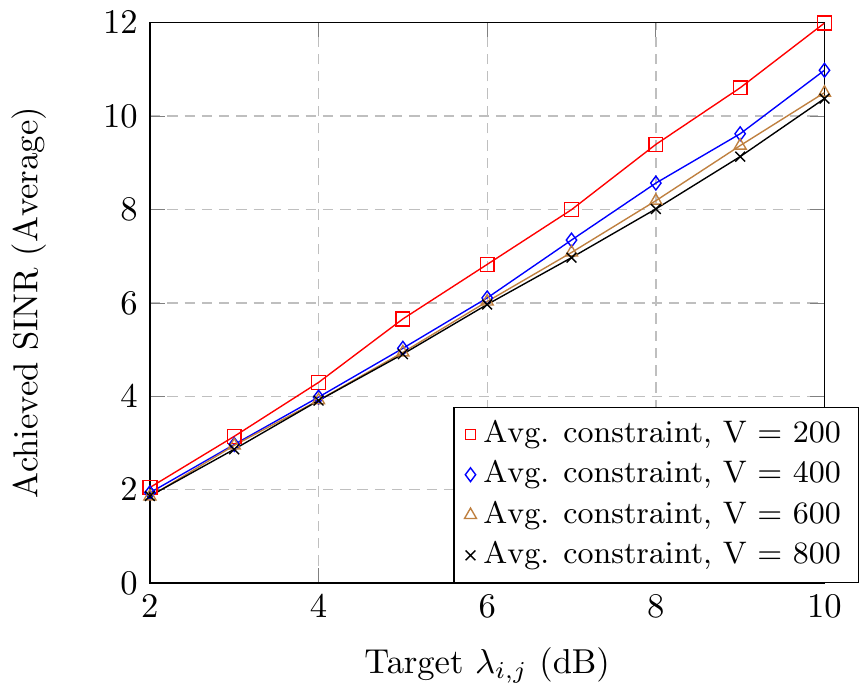}  
\caption{Achieved time average SINR Vs target QoS for a two cell scenario, each cell consisting of two UTs, $N_t = 5,$ $P_{\text{peak}} = 10$dB.}
\label{fig:feas_SINRVsTarget}
\end{scriptsize}
 \end{figure}

We next plot the time average energy expenditure per SCBS versus the average queue-length
for different values of $V$  obtained by running the DBF algorithm for $1000$ time slots in Figure \ref{fig:PowQLV}. 
The target time average QoS is $10$dB. It can be seen that 
as the value of $V$ increases, the time average energy expenditure
decreases, and the average queue-length increases. This is in accordance with the performance bounds of DBF algorithm.
Increasing $V$ implies that the SCBS transmits less frequently resulting in higher average queue-length  and lower average energy expenditure.
\begin{figure}[htp]
\begin{scriptsize}
 \includegraphics[width=3.5 in]{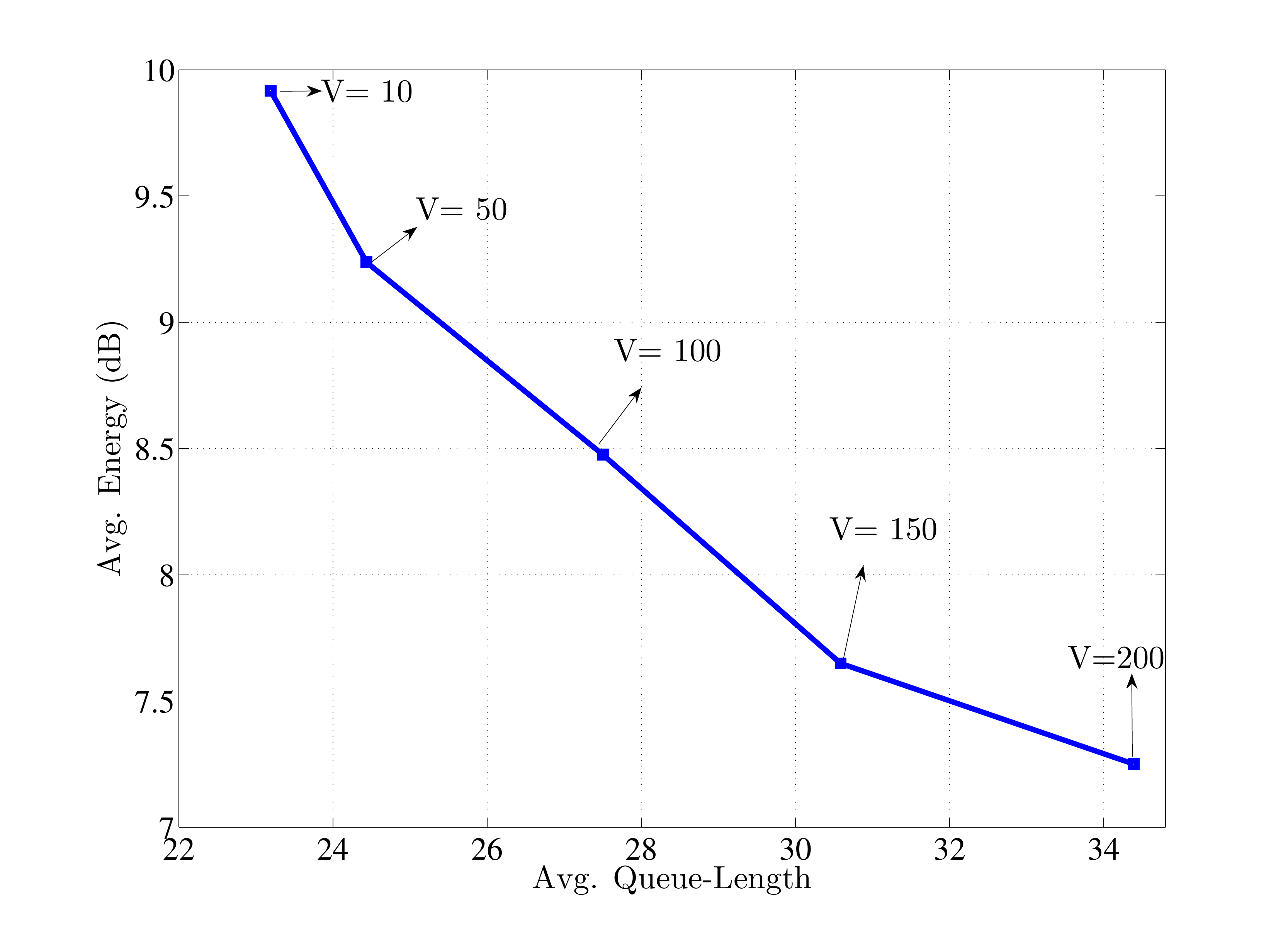}  
\caption{Time average queue-length vs time average energy expenditure for a two cell scenario, each cell consisting of two UTs, $N_t = 5,$ peak power per SCBS = $10$dB Target QoS value= $10$dB.}
\label{fig:PowQLV}
\end{scriptsize}
 \end{figure} 

Next, we examine the impact of the number of transmit antennas on the target QoS. We plot the average queue-length (of the virtual queue) as a function of the target QoS for different number of transmit antennas in Figure \ref{fig:QLVsTarget}. 
\begin{figure}[htp]
\begin{scriptsize}
 \includegraphics[width=3.5 in]{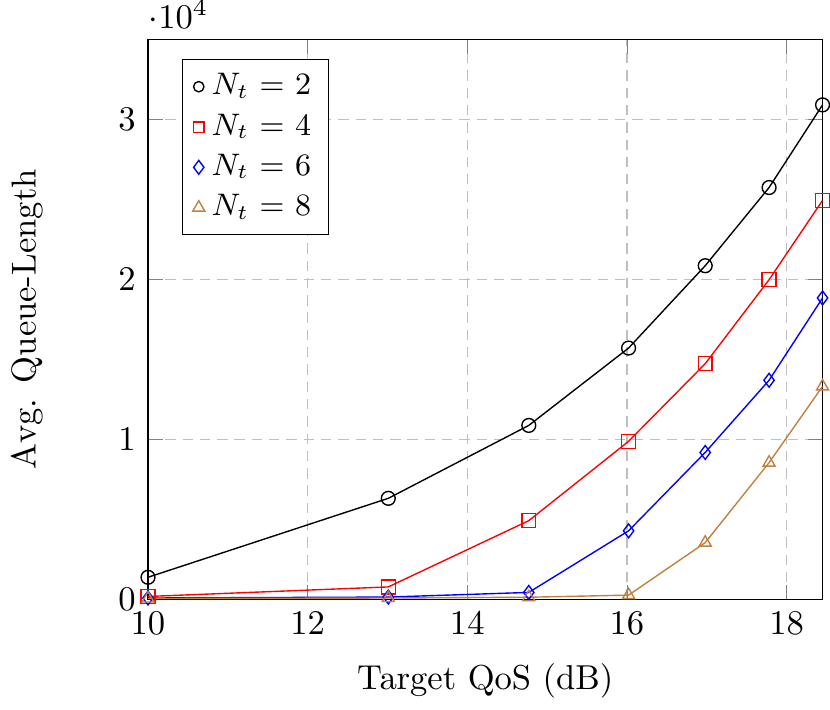}  
\caption{Time average queue-length vs target QoS in dB for different number of transmit antennas, peak power per SCBS = $10$dB, V = 100.}
\label{fig:QLVsTarget}
\end{scriptsize}
 \end{figure} 
First, it can be seen that as the number of transmit antennas increase, the average queue-length 
becomes lower. Also, it can be seen that there is a cut-off point beyond which the average queue-length blows up.
The cut off point represents the maximum supportable QoS target for the system.
The higher the number of transmit antennas, higher is the maximum supportable QoS.
This is due to the fact that higher number of antennas leads to greater degrees of freedom
resulting in enhancement of the useful signal power and less interference power.

 \section{Conclusion}
 \label{sec:Conclusion}
 In this work, we handled the problem of minimizing the transmit power expenditure subject to satisfying time
 average QoS constraints in a SCN scenario. Using the technique of Lyapunov optimization,  we proposed a decentralized online beamforming design algorithm whose performance in terms of the time average power expenditure can be made arbitrarily close to the optimal. Our results show that time average QoS constraints can lead to better savings in terms of transmit power as compared to solving the problem with instantaneous constraints. Additionally, we showed with the help of numerical results that the achieved time average SINR with our algorithm also satisfies the target value.
In addition we also considered two practical cases of interest in SCNs. In the first, we
considered the impact of delay in information exchange among the SCBSs. We showed that
the performance of the proposed algorithm with delays is only affected upto a finite constant in comparison with the case of no delays. Secondly, we considered the impact of CSI feedback.
We formulated a joint CSI feedback and beamforming framework using Lyapunov optimization technique. Furthermore, we provided a low complexity
algorithm to optimally solve the CSI feedback problem (that is obtained by the analysis of Lyapunov optimization).
We then provided performance bounds between our solution and the optimal solution of the original joint feedback and beamforming design stochastic problem.


\section*{Appendix A: Performance Bounds }
\label{app:PerfectCSIcase}
\subsection*{Part 1: Proof of Proposition \ref{prop:QLBoundm}}  
From \eqref{eqn:QLevolution_notations}, we can write the following.
 \begin{align}
Q^2_{i,j}[t+1] & \leq (Q_{i,j}[t]-\mu_{i,j}[t])^2+A^2_{i,j}[t] \nonumber \\ & +2A_{i,j}[t] \max \LB 0,Q_{i,j}[t]-\mu_{i,j}[t]\RB \nonumber \\
 & \leq Q^2_{i,j}[t] + \mu^2_{i,j}[t] + A^2_{i,j}[t] \nonumber \\ & - 2 Q_{i,j}[t] \LB \mu_{i,j}[t] - A_{i,j}[t] \RB.
\end{align}
Summing with respect to $i,j$ and taking the conditional expectation
$\mathbb{E}[.|\Qm[t]], $ we have,
\begin{align}
\Delta (\Qm[t]) & \leq \sum_{i,j} \mathbb{E} \LSB \mu^2_{i,j}[t] + A^2_{i,j}[t] | \Qm[t] \RSB \nonumber \\
&- \sum_{i,j} Q_{i,j} [t] \mathbb{E} \LSB \mu_{i,j}[t] - A_{i,j}[t] | \Qm[t] \RSB \label{eqn:here1}
\end{align}
Using the bound
$\mu_{i,j}[t] \leq \mu_{\max}$ 
and $A_{i,j}[t] \leq A_{\max},$
and defining $C_1 = \sum_{i,j}((A^{i,j}_{\max})^2+(\mu^{i,j}_{\max})^2).$ 
Taking the upper bound $A^{i,j}_{\max} \leq A_{\max}$ and $\mu^{i,j}_{\max} \leq \mu_{\max},$ 
we obtain $C_1 =  NK(A^2_{\max}+\mu^2_{\max}),$
we obtain the bound of \eqref{eqn:papa501m}.

\subsection{Part II : Proof of Proposition \ref{prop:DBFperformance}}
 From \eqref{eqn:papa51}, for the DBF policy we have,
 
 \begin{small}
 \begin{align}
& \Delta  (\Qm[t]) + V \mathbb{E} \big{[}  \sum_{i,j} (\wv^{\text{DBF}}_{i,j})^H [t]\wv^{\text{DBF}}_{i,j}[t] \big{|} \Qm[t] \big{]} \leq C_1 \nonumber \\ & \ +  \sum_{i,j}Q_{i,j}[t] (\lambda_{i,j}+\nu_{i,j} N_0)   \nonumber  \\ & - \sum_{i,j} \mathbb{E} \big{[} Q_{i,j}[t] \Big{(} (|\wv^{\text{DBF}}_{i,j} [t])^H\hv_{i,i,j} [t]|^2 \nonumber \\ & - \nu_{i,j} \sum_{\substack{(n,k) \\ \neq (i,j)}}|(\wv^{\text{DBF}}_{n,k} [t])^H \hv_{n,i,j} [t]|^2 \big{)} \nonumber \\ & \qquad -V  (\wv^{\text{DBF}}_{i,j} [t])^H\wv^{\text{DBF}}_{i,j}[t]\big{|} \Qm[t]\big{]} \label{eqn:papa300} \\
& \stackrel{(a)}{\leq} C_1 +  \sum_{i,j}Q_{i,j}[t] (\lambda_{i,j}+\nu_{i,j}    N_0) \nonumber \\ & \qquad - \sum_{i,j} \mathbb{E} \big{[} Q_{i,j}[t] \big{(} |\wv^{\text{TS}}_{i,j} [t])^H\hv_{i,i,j} [t]|^2 \nonumber \\ & \qquad \qquad -  \nu_{i,j} \sum_{\substack{(n,k) \\ \neq (i,j)}}|(\wv^{\text{TS}}_{n,k} [t])^H \hv_{n,i,j} [t]|^2 \big{)} \nonumber \\ & \qquad -V  (\wv^{\text{TS}}_{i,j} [t])^H \wv^{\text{TS}}_{i,j}[t]\big{|} \Qm[t]\big{]} \label{eqn:papa303}
\end{align}
 \end{small}
where the beamforming vector $\wv^{\text{TS}}_{i,j}$ is the one implemented with
any stationary randomized policy.
Inequality $(a)$ is true due to the following reason.
Recall that the DBF algorithm is implemented to maximize the RHS of the bound in \eqref{eqn:papa300}.
Therefore, replacing \eqref{eqn:papa300} with any other control policy should yield the inequality of $(a).$

In particular we replace it by a stationary randomized policy which satisfies the following
conditions.

\begin{small}
\begin{align}
& \mathbb{E}  \Big{[} |(\wv{^\text{TS}}_{i,j}[t])^H \hv_{i,i,j}[t]|^2 - \nonumber \\ &  \nu_{i,j} N_0 \sum_{\substack{(n,k) \\ \neq (i,j)}}  |(\wv{^\text{TS}}_{n,k}[t])^H \hv_{n,i,j}[t]|^2 -\nu_{i,j}N_0 \big{|} \Qm(t) \Big{]} \geq \lambda_{i,j} + \epsilon,    \label{eqn:papa301}\\
& \mathbb{E} \Big{[} \sum_{i,j} ({\wv}{^\text{TS}}_{i,j}[t])^H \wv{^\text{TS}}_{i,j}[t] \big{|} \Qm(t) \Big{]} = P_{\text{inf}}
\label{eqn:papa302}
\end{align}
\end{small}
The existence of such a policy is guaranteed by arguments
from Caratheodory's theorem and its proof is omitted
for brevity (the reader can refer to \cite{neelybook2006} for the details of the proof).
Using \eqref{eqn:papa301} and \eqref{eqn:papa302} in \eqref{eqn:papa303} yields,
\begin{align}
\Delta & (\Qm(t))   + V \mathbb{E} \big{[}  \sum_{i,j} (\wv^{\text{DBF}}_{i,j} [t])^H\wv^{\text{DBF}}_{i,j}[t] \big{|} \Qm[t] \big{]}   \nonumber \\ & \leq C_1+ \sum_{i,j} Q_{i,j} \lambda_{i,j} - \sum_{i,j} Q_{i,j} (\lambda_{i,j}+\epsilon) + V P_{\text{inf}}  \nonumber \\
&  = NK B_1 - \epsilon \sum_{i,j} Q_{i,j} + V P_{\text{inf}}  \label{eqn:papa880}
\end{align}
Using the result of \eqref{eqn:papa880}, and following some
straightforward steps (similar to Lemma 4.1, \cite{neelybook2006}, omitted here for brevity) it can be shown that,
\begin{align*}
& \lim_{T \to \infty} \frac{1}{T} \sum^{T-1}_{t = 0} \sum_{i,j} \Em \LSB Q^{\text{DBF}}_{i,j}[t]\RSB \leq \frac{C_1+V NK P_{\text{peak}}}{\epsilon} \\
& \lim_{T \to \infty} \frac{1}{T} \sum^{T-1}_{t = 0} \sum_{i} \Em \LSB P^{\text{DBF}}_i[t] \RSB
\leq P_{\inf}+\frac{C_1}{V}
\end{align*}

\section*{Appendix B : Performance with Delayed Queue-length exchange}
\label{app:delQueue}
\subsection{Part I : Proof of Lemma \ref{lem:QLDelBd}}
From the definitions of $\Wm^{\text{DBF}}_{i,j}[t]$ and $\Wm^\tau_{i,j}[t],$
we can conclude the following
\begin{align}
\sum_{j}\trace(\Am_{i,j}[t]\Wm^{\text{DBF}}_{i,j}[t]) & \geq \sum_{j} \trace(\Am_{i,j}[t] \Wm^{\text{del}}_{i,j}[t]) \label{eqn:papa9}\\
\sum_{j} \trace(\Am^\tau_{i,j}[t]\Wm^{\text{del}}_{i,j}[t]) & \geq \sum_{j}\trace(\Am^\tau_{i,j}[t] \Wm^{\text{DBF}}_{i,j}[t]) \label{eqn:papa10}
\end{align}
Recall the expression for  $\Am_{i,j} [t]$ given by
\begin{align*}
& \Am_{i,j} [t]  = Q_{i,j}[t] \Hm_{i,i,j}[t]-  \sum_{\substack{k \neq j}} \nu_{i,k}  Q_{i,k}[t] \Hm_{i,i,k}[t] 
\\& -  \sum_{\substack{n \neq i, k}} \nu_{n,k} Q_{n,k}[t] \Hm_{i,n,k}[t]  - V 
\end{align*}
Adding and subtracting $\sum_{\substack{n \neq i, k}} \nu_{n,k} Q_{n,k}[t-\tau] \Hm_{i,n,k}[t]$ on
the right hand side, we obtain
\begin{align}
\Am_{i,j} [t] & = Q_{i,j}[t] \Hm_{i,i,j}[t]- \sum_{\substack{k \neq j}} \nu_{i,k}  Q_{i,k}[t] \Hm_{i,i,k}[t]  \nonumber 
\\ & - \sum_{\substack{n \neq i, k}} \nu_{n,k} Q_{n,k}[t-\tau] \Hm_{i,n,k}[t] \nonumber 
\\& + \sum_{\substack{n \neq i, k}} \nu_{n,k} (Q_{n,k}[t-\tau]-Q_{n,k}[t]) \Hm_{i,n,k}[t]  - V  \nonumber 
\\ & = \Am^\tau_{i,j} [t] + \sum_{\substack{n \neq i, k}} \nu_{n,k} Q^d_{n,k}[t] \Hm_{i,n,k}[t] \label{eqn:papa11}
\end{align}  
where $Q^d_{n,k}[t] = Q_{n,k}[t-\tau]-Q_{n,k}[t].$
Using \eqref{eqn:papa11} in the inequality \eqref{eqn:papa9} yields,
\begin{align}
\sum_{j} & \trace(\Am^\tau_{i,j}[t] \Wm^{\text{DBF}}_{i,j}[t]) \nonumber \\
& +  \sum_{j} \sum_{\substack{n \neq i, k}} \nu_{n,k} Q^d_{n,k}[t] \trace( \Hm_{i,n,k}[t] \Wm^{\text{DBF}}_{i,j}[t]) \nonumber \\
& \geq \sum_{j} \trace(A^\tau_{i,j}[t]\Wm^{\text{del}}_{i,j}[t]) \nonumber \\
& +  \sum_{j} \sum_{\substack{n \neq i, k}} \nu_{n,k} Q^d_{n,k}[t] \trace( \Hm_{i,n,k}[t] \Wm^{\text{del}}_{i,j}[t]) \label{eqn:papa12}
\end{align}
Adding \eqref{eqn:papa10} to \eqref{eqn:papa12} yields,
\begin{align}
\sum_{j} \sum_{\substack{n \neq i, k}} \nu_{n,k} Q^d_{n,k}[t] \trace( \Hm_{i,n,k}[t] (\Wm^{\text{DBF}}_{i,j}[t]-\Wm^{\text{del}}_{i,j}[t]))  \geq 0. \label{eqn:papa13}
\end{align}
Finally consider the term $\sum_j \trace(\Am_{i,j}[t](\Wm^{\text{DBF}}_{i,j}[t]- \Wm^{\text{del}}_{i,j}[t])).$
Using \eqref{eqn:papa11}, we obtain
\begin{align}
 \sum_j & \trace(\Am_{i,j}[t](\Wm^{\text{DBF}}_{i,j}[t]- \Wm^{\text{del}}_{i,j}[t])) \nonumber\\ = & \sum_j \trace(\Am^\tau_{i,j}[t](\Wm^{\text{DBF}}_{i,j}[t]- \Wm^{\text{del}}_{i,j}[t])) \nonumber \\
 & +  \sum_j \sum_{\substack{n \neq i, k}} \nu_{n,k} Q^d_{n,k}[t] \trace( \Hm_{i,n,k}[t] (\Wm^{\text{DBF}}_{i,j}[t]-\Wm^{\text{del}}_{i,j}[t])) \nonumber \\
 &  \leq  \sum_j \sum_{\substack{n \neq i, k}} \nu_{n,k} Q^d_{n,k}[t] \trace( \Hm_{i,n,k}[t] (\Wm^{\text{DBF}}_{i,j}[t]-\Wm^{\text{del}}_{i,j}[t])) 
 \label{eqn:papa16}
\end{align}
where the last inequality follows due to the fact that (from \eqref{eqn:papa9}) 
\begin{align*}
\sum_j \trace(\Am^\tau_{i,j}[t](\Wm^{\text{DBF}}_{i,j}[t]- \Wm^{\text{del}}_{i,j}[t])) \leq 0.
\end{align*}
 Using the equation for queue-length evolution in
\eqref{eqn:QLevolution_notations} and the bounds in 
$A_{i,j}[t] \leq A_{\max}$ and $\mu_{i,j}[t] \leq \mu_{\max}$,
we can conclude the following.
\begin{align}
Q_{i,j} [t-\tau] + \tau A_{\max} & \geq Q_{i,j} [t] \qquad \forall i,j \nonumber \\
Q_{i,j} [t-\tau] - \tau \mu_{\max} & \geq Q_{i,j} [t] \qquad \forall i,j \label{eqn:papa15}
\end{align}
Combining the above equations yield,
\begin{align}
Q_{i,j} [t] - \tau A_{\max} \leq Q_{i,j} [t-\tau] \leq Q_{i,j} [t] + \tau \mu_{\max}, \  \forall i,j
\end{align}
And hence we can conclude that 
\begin{align}
 - \tau A_{\max} \leq Q^d_{i,j} [t] \leq \tau \mu_{\max} \qquad \forall i,j
\end{align}
Note that in \eqref{eqn:papa13}, we have already shown that the right hand side of \eqref{eqn:papa16} is positive.

In order to bound the right hand side of \eqref{eqn:papa16}, we proceed as follows.
First, we recall that there can only be one active UT per cell.
Therefore, only one of $\Wm^{\text{DBF}}_{i,j}$ is a non-zero matrix.
Similarly, the case for $\Wm^{\text{del}}_{i,j}.$
Therefore, 
\begin{align}
& \sum_j  \sum_{\substack{n \neq i, k}} \nu_{n,k} Q^d_{n,k}[t] \trace( \Hm_{i,n,k}[t] (\Wm^{\text{DBF}}_{i,j}[t]-\Wm^{\text{del}}_{i,j}[t])) \nonumber \\ 
& \leq 2 \max \{ \tau A_{\max},\tau \mu_{\max}\} P_{\text{peak}} \max_{n \neq i,k} |h_{i,n,k}[t]|^2  \defines B_1 
\end{align}
Therefore,
\begin{align*}
\sum_{j} \trace(\Am_{i,j}[t]\Wm^{\text{DBF}}_{i,j}[t]) \leq  \sum_{j} \trace(\Am_{i,j}[t]\Wm^{\text{del}}_{i,j}[t])+ B_1 \ \forall i,t.
\end{align*}
and hence,
\begin{align}
\sum_{i,j} \trace(\Am_{i,j}[t]\Wm^{\text{DBF}}_{i,j}[t]) \leq  \sum_{i,j} \trace(\Am_{i,j}[t]\Wm^{\text{del}}_{i,j}[t])+C_2,  \forall t.
\end{align}
where $C_2 = N B_1 < \infty.$

\subsection*{Part II : Proof of Theorem \ref{thm:DelQL}}
Rewriting \eqref{eqn:papa51}, we have
\begin{align}
& \Delta  (\Qm[t]) + V \mathbb{E} \big{[}  \sum_{i,j} \wv^H_{i,j} [t]\wv_{i,j}[t] \big{|} \Qm[t] \big{]} \leq  \nonumber \\
& = C_1 +  \sum_{i,j}Q_{i,j}[t] (\lambda_{i,j}+\nu_{i,j} N_0) -\sum_{i,j} \trace(\Am_{i,j}[t]\Wm_{i,j}[t]) 
\end{align}
where we have used the equivalence of the quadratic form and the trace
form.
With delayed queue-length information, the policy corresponding to $\Wm^{\text{del}}$ is implemented.
Therefore, we have

\begin{small}
\begin{align}
\Delta & (\Qm(t)) + V \mathbb{E} \big{[}  \sum_{i,j} (\wv^{\text{del}}_{i,j} [t])^H\wv^{\text{del}}_{i,j}[t] \big{|} \Qm[t] \big{]} \nonumber \\ 
& \leq C_1 +  \sum_{i,j}Q_{i,j}[t] (\lambda_{i,j} + \nu_{i,j} N_0) -\sum_{i,j} \trace(\Am_{i,j}[t]\Wm^{\text{del}}_{i,j}[t])  \nonumber \\
& \stackrel{(a)}{\leq} C_1 +  \sum_{i,j}Q_{i,j}[t] (\lambda_{i,j} + \nu_{i,j} N_0) \nonumber \\ & \qquad \qquad \qquad  - (\trace(\Am_{i,j}[t] \Wm^{\text{DBF}}_{i,j}[t]) - C_2) \nonumber \\
& = C_1 + C_2+ \sum_{i,j} Q_{i,j}[t] (\lambda_{i,j} +\nu_{i,j} N_0) - \trace(\Am_{i,j}[t] \Wm^{\text{DBF}}_{i,j}[t])
\end{align}
\end{small}
where (a) follows from Lemma \ref{lem:QLDelBd}.
Now following similar steps as Appendix A, part II, 
we can obtain the bounds of \eqref{eqn:QLBd_Del} and \eqref{eqn:PowBd_Del}.

\section*{Appendix C: CSI Feedback Scheme}

\subsection*{Part I : Proof of Theorem \ref{thm:QLFedBd}}
Let  $b^*_{i,n,k}$ be the solution obtained by our feedback algorithm. First note that $\sum_{n,k} b^*_{i,n,k} = B_{\max}$. Let $S^*_i$ be the set of UTs that feedback their CSI according to our algorithm to the SCBS$_i$ i.e. $b^*_{i,n,k}=1$ $\forall (n,k) \in S^*_i$. Let $\bar{S}^*_i$ the set of UTs such that $b^*_{i,n,k}=0$ $\forall (n,k) \in \bar{S}^*_i$. We prove the theorem by showing that if we replace any UT in $S^*_i$ by a UT in $\bar{S}^*_i,$ the resulting objective function of \eqref{eqn:FBcostfn} will have a smaller value. 

We examine the following 3 cases: i) both UTs are in the cell $i$, ii) both UTs belong to other cell ($n \neq i$) and iii) one UT belongs to cell $i$ and the second belongs to other cell ($n \neq i$). Lets use the notation $\Hm_{i,n,k}=\sigma_{i,n,k} \Hm^\prime_{i,n,k}$. One can see that $\Hm^\prime_{i,n,k}$ and $\Hm^\prime_{i,n,k^\prime}$ are independent of each other for $k\neq k^\prime,$ and $\bar{\Hm^\prime}_{i,n,k}=\mathbb{E} [ \Hm^\prime_{i,n,k}]=I$.

i) First, lets consider two UTs $j$ and $j^\prime$ in cell $i$ such that $b^*_{i,i,j}=1$ and $b^*_{i,i,j^\prime}=0$. Clearly $\tilde{Q}_{i,j}>\tilde{Q}_{i,j^\prime}$ (by the steps followed in our algorithm). Lets consider the objective function of \eqref{eqn:FBcostfn} i.e.

\begin{align}
 \mathbb{E} \big{[} \max_{\Pm} \sum_{i,j} P_{i,j}  \lambda^{\max}\Big{(} Q_{i,j} ( b_{i,i,j}\tilde{\Hm}_{i,i,j} + \bar{\Hm}_{i,i,j})  \nonumber \\ - \sum_{\substack{(n,k) \\ \neq (i,j)}} \nu_{n,k} Q_{n,k} {(}  b_{i,n,k} \tilde{\Hm}_{i,n,k} + \bar{\Hm}_{i,n,k}) - V \Id \big{)} \big{]}
  \end{align}
 and denote by matrix $C_{i}(\Hm_{ii_-})=\sum_{\substack{(n\neq i,k) }} \nu_{n,k} Q_{n,k} {(}  b_{i,n,k} \tilde{\Hm}_{i,n,k} + \bar{\Hm}_{i,n,k}) + V \Id$ where $\Hm_{ii_-}$ is the matrix that contains the channels between the base station $i$ and all UTs in other cells $n\neq i$. In order to simplify the proof description, lets consider that there exists a third UT $k$ in cell $i$ that lies in set $S^*_i$ (however our argument holds for any number of UTs in the system). One can notice that for every channel realization, the solution  corresponding to the optimal beamforming and power allocation in \eqref{eqn:FBcostfn} is such that only one UT is active, or none of the UTs are active. Therefore, $P_{i,j} $ is either 
$P_{peak}$ or 0. According to our feedback algorithm $b^*_{i,i,j}=1$ and the objective function in \eqref{eqn:FBcostfn} corresponding to UT $j$ is  
\begin{align*}
F_{i,j} &= P_{i,j}  \lambda^{\max} ( \tilde{Q}_{i,j} (\Hm^\prime_{i,i,j})  \nonumber \\ &- C_{i}(\Hm_{ii_-}) - \nu_{i,k} \tilde{Q}_{i,k} \Hm^\prime_{i,i,k} - \nu_{i,j^\prime} \tilde{Q}_{i,j^\prime} \bar{\Hm^\prime}_{i,i,j^\prime}) .
\end{align*}
$F_{i,j}\geq 0$ and $F_{i,j}\geq F_{i,k}$ for some channel and virtual queue states i.e. setting $b^*_{i,i,j}=1$ will increase the objective function of \eqref{eqn:FBcostfn}. 
If we interchange UTs $j$ and $j^\prime$ i.e. we set $b^*_{i,i,j}$ to 0 and $b^*_{i,i,j^\prime}$ to 1,  the objective function in \eqref{eqn:FBcostfn} corresponding to UT $j^\prime$ is 

\begin{align}
F_{i,j^\prime} &= P_{i,j^\prime}  \lambda^{\max} (\tilde{Q}_{i,j^\prime} (\Hm^\prime_{i,i,j^\prime})  \nonumber \\ &- C_{i}(\Hm_{ii_-}) - \nu_{i,k} \tilde{Q}_{i,k} \Hm^\prime_{i,i,k} - \nu_{i,j} \tilde{Q}_{i,j} \bar{\Hm^\prime}_{i,i,j} )  \nonumber \\ & \stackrel{(a)}{=} P_{i,j^\prime}  \lambda^{\max} ( \tilde{Q}_{i,j^\prime} (\Hm^\prime_{i,i,j^\prime}) - C_{i}(\Hm_{ii_-}) - \nu_{i,k} \tilde{Q}_{i,k} \Hm^\prime_{i,i,k} )  \nonumber \\ &- P_{i,j^\prime} \nu_{i,j} \tilde{Q}_{i,j} \nonumber \\   & \stackrel{(b)}{\leq} P_{i,j^\prime}  [ \lambda^{\max} {(} \tilde{Q}_{i,j^\prime} (\Hm^\prime_{i,i,j^\prime}) {)}-\nu_{i,j} \tilde{Q}_{i,j} ] \stackrel{(c)}{\leq} 0.
\end{align}
Therefore $P_{i,j^\prime}$ is set to 0 (in order  to maximize the objective function of \eqref{eqn:FBcostfn}) and $F_{i,j^\prime}=0$ for all channel states.  Notice that (a) follows from $\bar{\Hm^\prime}_{i,i,j}=I$ and (b) follows from  $\wv^H_{i,j^\prime}( C_{i}(\Hm_{ii_-}) + \nu_{i,k} \tilde{Q}_{i,k} \Hm^\prime_{i,i,k} )\wv_{i,j^\prime} \geq 0$ $\forall$ $\wv_{i,j^\prime}$. (c) follows from $\tilde{Q}_{i,j}>\tilde{Q}_{i,j^\prime}$, $ \lambda^{\max}(\Hm^\prime_{i,i,j^\prime})=1$ and $\nu_{i,j}>1$ (target SNR is higher than 0 dB). Consequently, setting $b^*_{i,i,j}$ to 0 and $b^*_{i,i,j^\prime}$ to 1 will reduce the objective function of \eqref{eqn:FBcostfn}. 

ii) Lets now consider two UTs not belong to a cell $i$, denoted by the index $(n,k)$ and $(n^\prime,k^\prime)$. UT $(n,k)$ $\in$ $S^*_i$  i.e. $b^*_{i,n,k}=1$ and $(n^\prime,k^\prime)$ $\in$ $\bar{S}^*_i$  i.e. $b^*_{i,n^\prime,k^\prime}=0$. Therefore $\nu_{n^\prime,k^\prime} \tilde{Q}_{n^\prime,k^\prime} < \nu_{n,k} \tilde{Q}_{n,k}$.  
The objective function of any UT $j$ belonging to cell $i$ is 
\begin{align}
& F_{i,j} = P_{i,j}  \lambda^{\max} ( \tilde{Q}_{i,j} (\Hm^\prime_{i,i,j})  \nonumber \\ &- C_{i,j}(\Hm_{ij_-}) - \nu_{n,k} \tilde{Q}_{n,k} \Hm^\prime_{i,n,k} - \nu_{n^\prime,k^\prime} \tilde{Q}_{n^\prime,k^\prime} \bar{\Hm^\prime}_{i,n^\prime,k^\prime}) \nonumber \\ & \stackrel{(a)}{=} P_{i,j} [ \lambda^{\max} {(} \tilde{Q}_{i,j} (\Hm^\prime_{i,i,j})  - C_{i,j}(\Hm_{ij_-}) - \nu_{n,k} \tilde{Q}_{n,k} \Hm^\prime_{i,n,k} ) \nonumber \\ & \qquad \qquad - \nu_{n^\prime,k^\prime} \tilde{Q}_{n^\prime,k^\prime} ]
\end{align}
where matrix $C_{i,j}(\Hm_{ij_-})=\sum_{\substack{(l,k) }} \nu_{l,k} Q_{l,k} {(}  b_{i,l,k} \tilde{\Hm}_{i,l,k} + \bar{\Hm}_{i,l,k}) + V \Id$ where the sum is over all UTs in all cells $(l,k)$ (including cell $i$) except UTs $(i,j)$, $(n,k)$ and $(n^\prime,k^\prime)$.  $\Hm_{ii_-}$ is then the matrix that contains the channels between the base station $i$ and all UTs in other cells except channels $\Hm_{i,i,j}$, $\Hm_{i,n,k}$ and $\Hm_{i,n^\prime,k^\prime}$. Notice that (a) follows from $\bar{\Hm^\prime}_{i,n^\prime,k^\prime}=I$. The objective function of \eqref{eqn:FBcostfn} can be written as follows
\begin{align*}
 \mathbb{E}_{\substack{\Hm_{ij_-},\Hm^\prime_{i,i,j}}}  \mathbb{E}_{\substack{ \Hm^\prime_{i,n,k}}} \big{[} \max_{\Pm} \sum_{i,j} F_{i,j} \quad  \big{|} \Hm_{ij_-}, \Hm^\prime_{i,i,j} \big{]} 
  \end{align*}
If we change the feedback decision by setting $b^*_{i,n,k}=0$ and $b^*_{i,n^\prime,k^\prime}=1$, the objective function of UT $j$ becomes, 
\begin{align*}
F^\prime_{i,j} &=P_{i,j} \big{[}  \lambda^{\max}( \tilde{Q}_{i,j} (\Hm^\prime_{i,i,j})  - C_{i,j}(\Hm_{ij_-}) \\ & \qquad -\nu_{n^\prime,k^\prime} \tilde{Q}_{n^\prime,k^\prime}  \Hm^\prime_{i,n^\prime,k^\prime} ) \nonumber  - \nu_{n,k} \tilde{Q}_{n,k} \big{]}
\end{align*}
The objective function of \eqref{eqn:FBcostfn} becomes,
\begin{align}
 \mathbb{E}_{\substack{\Hm_{ij_-},\Hm^\prime_{i,i,j}}}  \mathbb{E}_{\substack{ \Hm^\prime_{i,n^\prime,k^\prime}}} \big{[} \max_{\Pm} \sum_{i,j} F^\prime_{i,j} \quad  \big{|} \Hm_{ij_-}, \Hm^\prime_{i,i,j}\big{]}
  \end{align}
Let $F_{i,j}(\Hm) $ and $F^\prime_{i,j}(\Hm)$ be the objective function of UT $(i,j)$ when respectively $ \Hm^\prime_{i,n,k}=\Hm$  and  $ \Hm^\prime_{i,n^\prime,k^\prime}=\Hm$. By definition of $\lambda^{\max}$, $F^\prime_{i,j}(\Hm) $ is given as,
\begin{align}
F^\prime_{i,j} & (\Hm)  =P_{i,j} \Big{[} \max_{\substack{||\wv_{i,j}||^2=1}} \wv^H_{i,j} ( \tilde{Q}_{i,j} (\Hm^\prime_{i,i,j})  - C_{i,j}(\Hm_{ij_-}) \nonumber \\ & -\nu_{n^\prime,k^\prime} \tilde{Q}_{n^\prime,k^\prime}  \Hm ) \wv_{i,j}  - \nu_{n,k} \tilde{Q}_{n,k} \Big{]} \nonumber \\ &= P_{i,j} \Big{[}   {\wv^{\prime*H}}_{i,j} ( \tilde{Q}_{i,j} (\Hm^\prime_{i,i,j})  - C_{i,j}(\Hm_{ij_-})  ) \wv^{\prime*}_{i,j} \nonumber \\ & -\nu_{n^\prime,k^\prime} \tilde{Q}_{n^\prime,k^\prime} {\wv^{\prime*H}}_{i,j}( \Hm ) \wv^{\prime*}_{i,j}  - \nu_{n,k} \tilde{Q}_{n,k} \big{]} \nonumber \\ & \stackrel{(a)}{=}  P_{i,j} \Big{[}   {\wv^{\prime*}}^H_{i,j} ( \tilde{Q}_{i,j} (\Hm^\prime_{i,i,j})  - C_{i,j}(\Hm_{ij_-})  ) \wv^{\prime*}_{i,j} \nonumber \\ & -\nu_{n,k} \tilde{Q}_{n,k}  {\wv^{\prime*H}}_{i,j}( \Hm ) \wv^{\prime*}_{i,j}  - \nu_{n^\prime,k^\prime} \tilde{Q}_{n^\prime,k^\prime} \nonumber \\ &+  {(}1- {\wv^{\prime*H}}_{i,j}( \Hm ) \wv^{\prime*}_{i,j} {)} {(}\nu_{n^\prime,k^\prime} \tilde{Q}_{n^\prime,k^\prime}  -\nu_{n,k} \tilde{Q}_{n,k}{)}\Big{]} \nonumber \\ &\stackrel{(b)}{\leq} P_{i,j} \Big{[}   {\wv^{\prime*}}^H_{i,j} {(} \tilde{Q}_{i,j} (\Hm^\prime_{i,i,j})  - C_{i,j}(\Hm_{ij_-})  {)} \wv^{\prime*}_{i,j} \nonumber \\ & -\nu_{n,k} \tilde{Q}_{n,k}  {\wv^{\prime*H}}_{i,j}( \Hm ) \wv^{\prime*}_{i,j}  - \nu_{n^\prime,k^\prime} \tilde{Q}_{n^\prime,k^\prime} \Big{]} \nonumber \\ &\leq P_{i,j} \Big{[}  \max_{\substack{||\wv_{i,j}||^2=1}} \wv^H_{i,j} ( \tilde{Q}_{i,j} (\Hm^\prime_{i,i,j})  - C_{i,j}(\Hm_{ij_-}) \nonumber \\ & -\nu_{n,k} \tilde{Q}_{n,k}  \Hm {)} \wv_{i,j}  - \nu_{n^\prime,k^\prime} \tilde{Q}_{n^\prime,k^\prime}  \Big{]} \stackrel{(c)}{=} F_{i,j}(\Hm) 
\end{align}
where (a) follows by adding and subtracting $\nu_{n^\prime,k^\prime} \tilde{Q}_{n^\prime,k^\prime}$ and $\nu_{n,k} \tilde{Q}_{n,k}  {\wv^{\prime*H}}_{i,j}( \Hm ) \wv^{\prime*}_{i,j}$. (b) follows from $\nu_{n^\prime,k^\prime} \tilde{Q}_{n^\prime,k^\prime} < \nu_{n,k} \tilde{Q}_{n,k}$ (due to our feedback algorithm since $(n,k)$ $\in$ $S^*_i$  and $(n^\prime,k^\prime)$ $\in$ $\bar{S}^*_i$) and ${\wv^{\prime*H}}_{i,j}( \Hm ) \wv^{\prime*}_{i,j} \leq 1$ ( $\lambda^{\max}(\Hm^\prime_{i,n,k}=\Hm)=1$ since $\Hm_{i,n,k}=\sigma_{i,n,k} \Hm^\prime_{i,n,k}$ where $\sigma_{i,n,k}$ is the average channel gain). (c) follows from the definition of $\lambda^{\max}$.
Therefore, for a given channel state $\Hm$ and for any UT $(i,j)$ we obtain $F^\prime_{i,j}(\Hm)  \leq F_{i,j}(\Hm)$. Using the fact that  $ \Hm^\prime_{i,n^\prime,k^\prime}$ and $ \Hm^\prime_{i,n,k}$ are i.i.d, we get 
\begin{align}
\mathbb{E}_{\substack{ \Hm^\prime_{i,n^\prime,k^\prime}}} \big{[} \max_{\Pm} \sum_{i,j} F^\prime_{i,j} \quad  \big{|} \Hm_{ij_-}, \Hm^\prime_{i,i,j}\big{]} \nonumber \\ \leq  \mathbb{E}_{\substack{ \Hm^\prime_{i,n,k}}} \big{[} \max_{\Pm} \sum_{i,j} F_{i,j} \quad  \big{|} \Hm_{ij_-}, \Hm^\prime_{i,i,j}\big{]}
  \end{align}
which implies,
\begin{align}
 \mathbb{E}_{\substack{\Hm_{ij_-},\Hm^\prime_{i,i,j}}} \mathbb{E}_{\substack{ \Hm^\prime_{i,n^\prime,k^\prime}}} \big{[} \max_{\Pm} \sum_{i,j} F^\prime_{i,j} \quad  \big{|} \Hm_{ij_-}, \Hm^\prime_{i,i,j}\big{]} \nonumber \\ \leq   \mathbb{E}_{\substack{\Hm_{ij_-},\Hm^\prime_{i,i,j}}} \mathbb{E}_{\substack{ \Hm^\prime_{i,n,k}}} \big{[} \max_{\Pm} \sum_{i,j} F_{i,j} \quad  \big{|} \Hm_{ij_-}, \Hm^\prime_{i,i,j}\big{]}
  \end{align}
Consequently, changing our feedback allocation $b^*_{i,n,k}$ will reduce the objective function of  \eqref{eqn:FBcostfn}.  

iii) To complete the proof, lets consider two UTs $(i,j)$ and $(n,k)$ where $n\neq i$. We assume that according to our algorithm  $b^*_{i,i,j}=1$  and  $b^*_{i,n,k}=0$ (the other case of $b^*_{i,i,j}=0$  and  $b^*_{i,n,k}=1$ can also be deduced in the same way). If we change the allocation by setting $b^*_{i,n,k}=1$ and $b^*_{i,i,j}=0$, the resulting objective function of   \eqref{eqn:FBcostfn} will be smaller. This is due to the following: According to i) our algorithm selects the best UTs in cell $i$ and according to ii) our algorithm selects the best UTs in other cells. Furthermore, in the second step of our algorithm, we compare between the selected UT in cell $i$ and the selected one from other cells and select the UT that maximizes the objective function of \eqref{eqn:FBcostfn}.

\subsection*{Part II : Proof of Corollary \ref{corol:QLFedBd}}

Recall the expression for Lyapunov drift from Proposition \ref{prop:QLBoundm}.
Recall that the expectation on the right hand side
of \eqref{eqn:papa51} is taken across the random processes in the system (and
in particular the randomness of the channel states).
When the SCBS has only the estimate of the channel state,
by the law of iterated expectations, we have

\begin{small}
\begin{align}
& \Delta_V \leq C_1+\sum_{i,j} Q_{i,j} (\lambda_{i,j} + \nu_{i,j} N_0) +
   \sum_{i,j} \mathbb{E}_{\Hm}   \Big{[} Q_{i,j}    |\wv^H_{i,j}{\hv}_{i,i,j}|^2  \nonumber \\ &  \qquad -  Q_{i,j} \nu_{i,j} \sum_{\substack{(n,k) \\ \neq (i,j)}}   |\wv^H_{n,k}{\hv}_{n,i,j}|^2   - V    \wv^H_{i,j} \wv_{i,j}   \big{|} \hat{\Hm} \Big{]} \nonumber \\
& = C+\sum_{i,j} Q_{i,j} (\lambda_{i,j}+\nu_{i,j} N_0) - \mathbb{E} \Big{[} \sum_{i,j}\wv^H_{i,j} {\Bm_{i,j}  \wv_{i,j}} \Big{]}
\end{align}
\end{small}
where $C_1$ is given as in Proposition~\ref{prop:QLBoundm} and the matrix 
$\Bm_{i,j} = Q_{i,j} \hat{\Hm}_{i,i,j}- \sum_{\substack{(n,k) \\ \neq (i,j)}} \nu_{n,k} Q_{n,k} \hat{\Hm}_{i,n,k}-V \Id,$
based on the estimate of the CSI.
For the feedback policy considered in this work, 
recall that if the UT$_{n,k}$
feeds back its CSI to the SCBS$_i,$
then $\hat{\hv}_{i,n,k} = \hv_{i,n,k}$
and if a UT$_{n,k}$
does not feed back its CSI to the SCBS$_i,$
then $\hat{\hv}_{i,n,k} = \mathbb{E}[\hv_{i,n,k}].$
Therefore, the matrix $\Bm_{i,j}$ can be compactly 
written in terms of the feedback indicator $b_{i,n,k}$
as
$ \Bm_{i,j} =  Q_{i,j} ( b_{i,i,j} \tilde{\Hm}_{i,i,j}+\bar{\Hm}_{i,i,j}) -  \sum_{\substack{(n,k) \\ \neq (i,j)}} \nu_{n,k} Q_{n,k} ( b_{i,n,k}\tilde{\Hm}_{i,n,k}+\bar{\Hm}_{i,n,k} ) - V \Id .$

As before, following the approach of Lyapunov optimization, 
we minimize the bound on the Lyapunov drift.
Therefore, for a given CSI feedback strategy 
$\bv,$ we consider the beamforming vector to maximize 
the term $\max_{\wv} \sum_{i,j} \wv^H_{i,j} \Bm_{i,j}  \wv_{i,j} $ and then
choose the feedback strategy to maximize the $\mathbb{E}[\max_{\wv} \sum_{i,j} \wv^H_{i,j} \Bm_{i,j}  \wv_{i,j}]$ and examine the performance of our algorithm.
Therefore we have,
\begin{align}
& \Delta_V  \leq C_1+\sum_{i,j} Q_{i,j} (\lambda_{i,j}+ \nu_{i,j} N_0) \nonumber \\
& \qquad \qquad - \max_{\bv} \mathbb{E} \big{[}\max_{\wv}\sum_{i,j}\wv^H_{i,j} \Bm_{i,j} \wv_{i,j} \big{]} \nonumber \\ & =  C_1+\sum_{i,j} Q_{i,j} (\lambda_{i,j}+ \nu_{i,j} N_0) \nonumber \\ & \qquad \qquad- \max_{\bv} \mathbb{E} \big{[} \max_{\Pm}\sum_{i,j} P_{i,j}\lambda^{\max}(\Bm_{i,j} ) \big{]} \label{eqn:uaua1}
\end{align}
where \eqref{eqn:uaua1} follows from steps similar to that of Proposition 2.

Using Theorem \ref{thm:QLFedBd}, our feedback and allocation policy 
minimizes the right hand side of  \eqref{eqn:uaua1} and therefore minimizes the 
bound on the Lyapunov function.  Replacing this by some other alternate feedback, beamforming and power allocation
policy $b^s_{i,n,k}$, $\vv^s_{i,j}$ and $P^s_{i,j},$  we have,
\begin{align}
& \Delta_V  \leq C_1+\sum_{i,j} Q_{i,j} (\lambda_{i,j}+\nu_{i,j} N_0)
+ V \mathbb{E} [  \sum_{i,j} P^s_{i,j}]  \nonumber \\
& \mathbb{E}\Big{[} -\sum_{i,j}  P^s_{i,j} \tilde{Q}_{i,j} {\vv^s_{i,j}}^H   (b^s_{i,i,j}\tilde{\Hm}_{i,i,j} + \bar{\Hm}_{i,i,j})  {\vv^s_{i,j}} \nonumber \\ & + \sum_{i,j} P^s_{i,j} \sum_{\substack{(n,k) \\ \neq (i,j)}}  \nu_{n,k}  \tilde{Q}_{n,k}  {\vv^s_{i,j}}^H (b^s_{i,n,k}\tilde{\Hm}_{i,n,k} + \bar{\Hm}_{i,n,k}) \vv^s_{i,j} \Big{]}
\end{align}
In particular choosing a stationary randomized policy such
that 
\begin{align}
& \mathbb{E}\Big{[} -\sum_{i,j}  P^s_{i,j} \tilde{Q}_{i,j} {\vv^s_{i,j}}^H   (b^s_{i,i,j}\tilde{\Hm}_{i,i,j} + \bar{\Hm}_{i,i,j})  {\vv^s_{i,j}} \nonumber \\ & + \sum_{i,j} P^s_{i,j} \sum_{\substack{(n,k) \\ \neq (i,j)}}  \nu_{n,k}  \tilde{Q}_{n,k}  {\vv^s_{i,j}}^H (b^s_{i,n,k}\tilde{\Hm}_{i,n,k} + \bar{\Hm}_{i,n,k}) \vv^s_{i,j} \Big{]} \nonumber \\ & \qquad \qquad \qquad \geq \lambda_{i,j} + \epsilon \\
& \mathbb{E} [  \sum_{i,j} P^s_{i,j}]  = P^\prime_{\text{inf}}(\epsilon)
\end{align}
we have
\begin{align}
& \Delta_V  \leq C_1- \epsilon \sum_{i,j} Q_{i,j} 
+ V P^\prime_{\text{inf}} \label{eqn:papa0880}
\end{align}
From \eqref{eqn:papa0880} and following some
straightforward steps (similar to Lemma 4.1, \cite{neelybook2006}, omitted here for brevity), we can conclude the result of Corollary 6.

\bibliographystyle{IEEEtran}
\bibliography{IEEEabrv,bibliography}

\begin{IEEEbiography}
[{\includegraphics[width=1in,height=1.25in,clip,keepaspectratio]{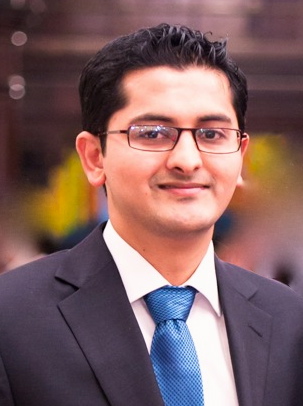}}]
{Subhash Lakshminarayana}
(S'07-M'12)
received the M.S. degree in electrical and computer engi- neering from The Ohio State University, Columbus, OH, USA, in 2009 and the Ph.D. degree from \'Ecole Sup\'erieure d'\'Electricit\'e (Sup\'elec), France in 2012.

Currently, he is with the Advanced Digital Sciences Center, Illinois at Singapore. In 2007, he was a Student Researcher with the Indian Institute of Science, Bangalore, India. From August 2013 to December 2013 and from May 2014 to November
2014, he held visiting research appointments with Princeton University. His research interests include the areas of smart grids and the security of cyber- physical systems, as well as wireless communication and signal processing, with emphasis on small-cell networks, cross-layer design wireless networks, MIMO systems, stochastic network optimization, and energy harvesting. He has served as a TPC member for top international conferences. His works have been selected among the Best conference papers on integration of renewable and intermittent resources at the IEEE PESGM - 2015 conference and the `Best 50 papers" of IEEE Globecom 2014 conference.
\end{IEEEbiography}

\begin{IEEEbiography}
[{\includegraphics[width=1in,height=1.25in,clip,keepaspectratio]{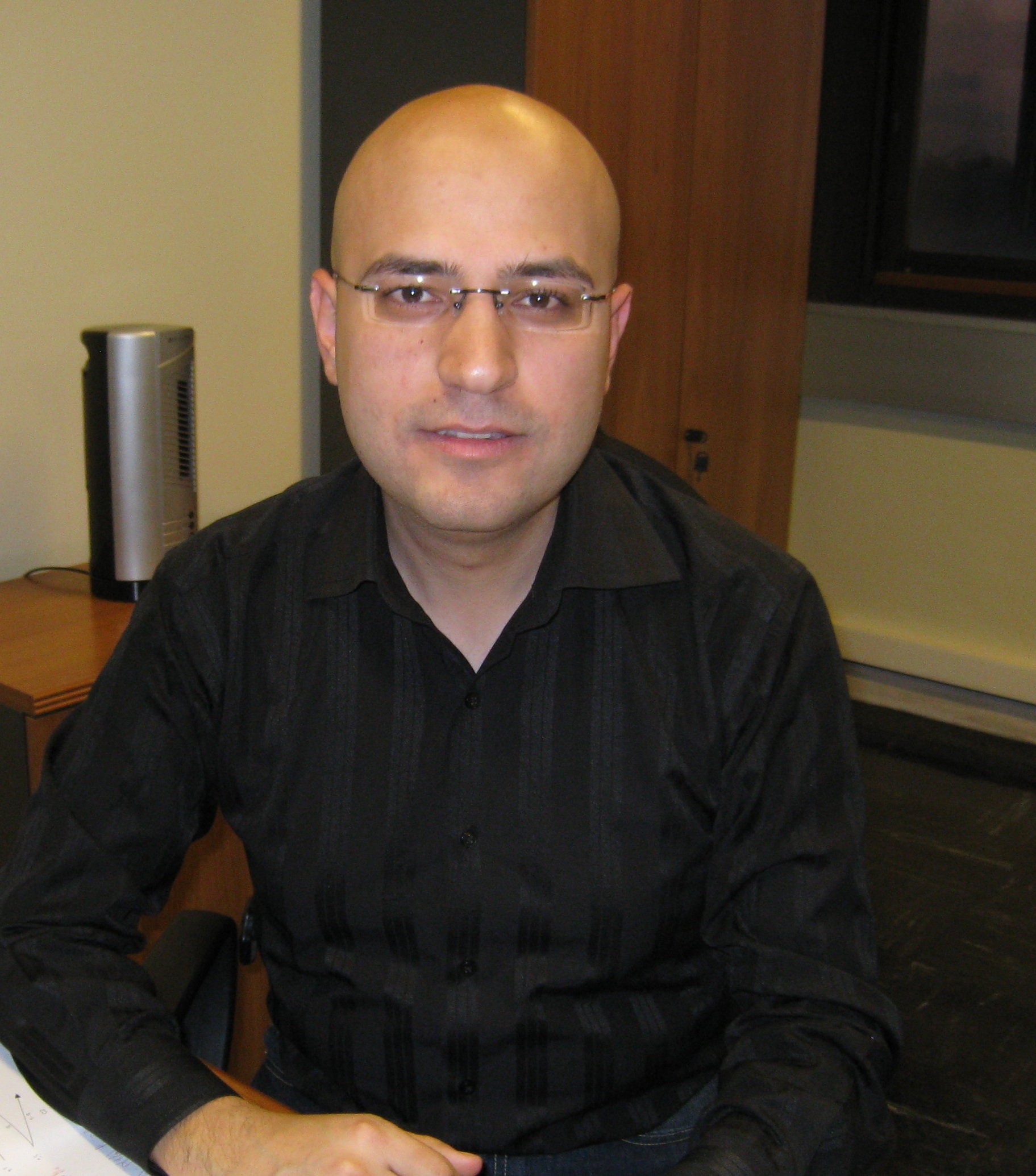}}]
{Mohamad Assaad}
received the BE degree in electrical engineering with high honors from Lebanese University, Beirut, in 2001, and the MSc and PhD degrees (with high honors), both in telecommunications, from Telecom ParisTech, Paris, France, in 2002 and 2006, respectively. Since March 2006, he has been with the Telecommunications Department at CentraleSup\'elec, where he is currently an associate professor. He has co-authored 1 book and more that 75 publications in journals and conference proceedings and serves regularly as TPC member for several top international conferences. His research interests include mathematical models of communication networks, resource optimization and cross-layer design in wireless networks, and stochastic network optimization.

\end{IEEEbiography}

\begin{IEEEbiography}
[{\includegraphics[width=1in,height=1.25in,clip,keepaspectratio]{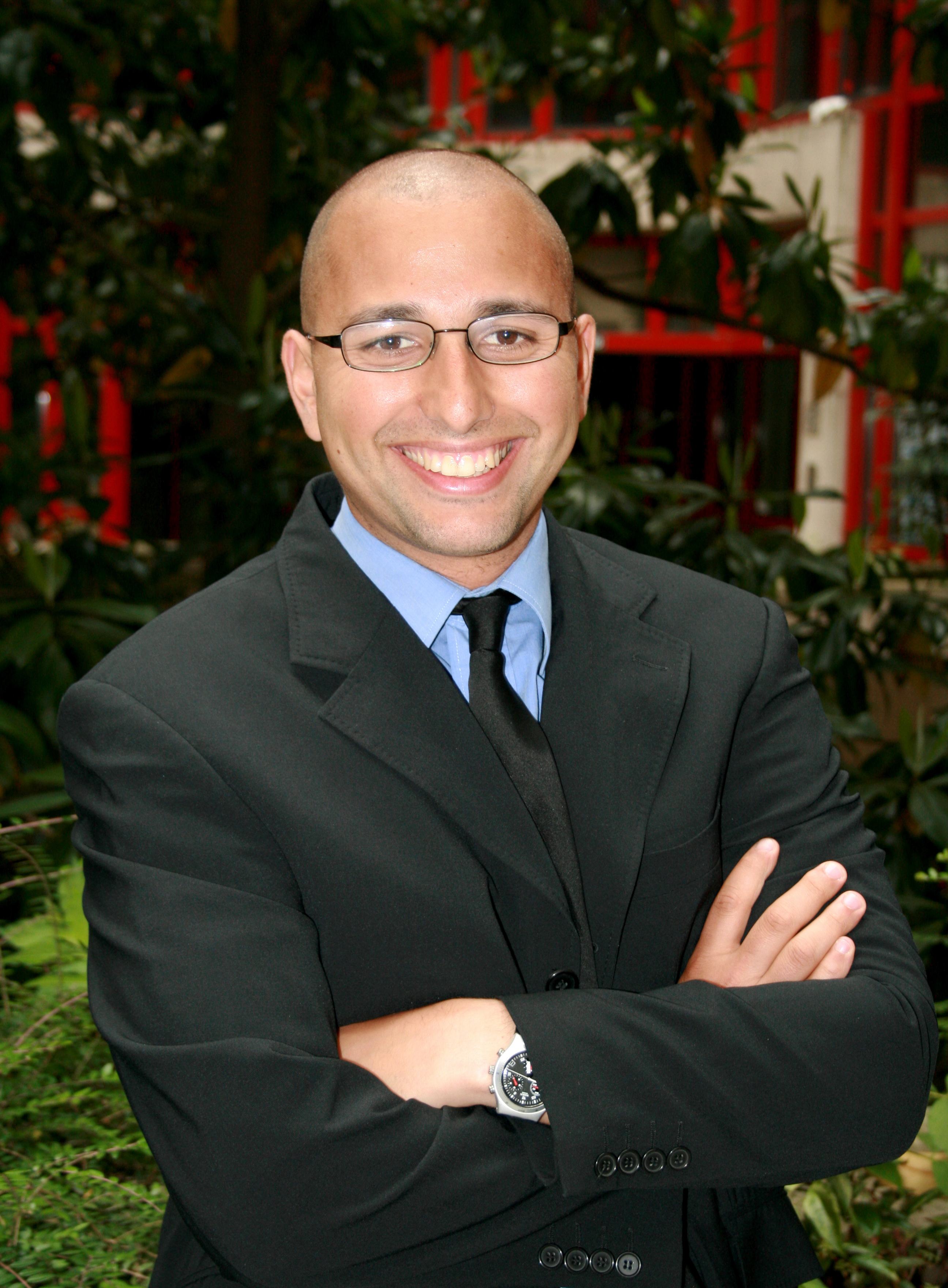}}]
{M\'erouane Debbah}
entered the \'Ecole Normale Sup\'erieure de Cachan (France) in 1996 where he received his M.Sc and Ph.D. degrees respectively. He worked for Motorola Labs (Saclay, France) from 1999-2002 and the Vienna Research Center for Telecommunications (Vienna, Austria) until 2003. From 2003 to 2007, he joined the Mobile Communications department of the Institut Eurecom (Sophia Antipolis, France) as an Assistant Professor. Since 2007, he is a Full Professor at Sup\'elec (Gif-sur-Yvette, France). From 2007 to 2014, he was director of the Alcatel-Lucent Chair on Flexible Radio. Since 2014, he is Vice-President of the Huawei France R\&D center and director of the Mathematical and Algorithmic Sciences Lab. His research interests are in information theory, signal processing and wireless communications. He is an Associate Editor in Chief of the journal Random Matrix: Theory and Applications and was an associate and senior area editor for IEEE Transactions on Signal Processing respectively in 2011-2013 and 2013-2014. M\'erouane Debbah is a recipient of the ERC grant MORE (Advanced Mathematical Tools for Complex Network Engineering). He is a IEEE Fellow, a WWRF Fellow and a member of the academic senate of Paris-Saclay. He is the recipient of the Mario Boella award in 2005, the 2007 IEEE GLOBECOM best paper award, the Wi-Opt 2009 best paper award, the 2010 Newcom++ best paper award, the WUN CogCom Best Paper 2012 and 2013 Award, the 2014 WCNC best paper award as well as the Valuetools 2007, Valuetools 2008, CrownCom2009 , Valuetools 2012 and SAM 2014 best student paper awards. In 2011, he received the IEEE Glavieux Prize Award and in 2012, the Qualcomm Innovation Prize Award.

\end{IEEEbiography}

\end{document}